\documentclass[a4paper,reqno]{amsart}
\usepackage[english]{babel}
\usepackage{amsmath,amssymb,amsfonts,amsthm}
\usepackage{hyperref}
\usepackage{slashed}
\usepackage{graphicx,color}
\usepackage{mathtools}
\usepackage{tikz-cd}
\usepackage[a4paper]{geometry}
\usepackage{csquotes} 

\theoremstyle{definition} 
\newtheorem{definition}{Definition}

\theoremstyle{plain} 
\newtheorem{theorem}[definition]{Theorem}

\newtheorem{lemma}[definition]{Lemma}
\newtheorem{corollary}[definition]{Corollary}

\theoremstyle{remark} 
\newtheorem{remark}[definition]{Remark}

\usepackage[bibstyle=numeric,citestyle=numeric,useprefix,backend=biber,giveninits=true]{biblatex}
\renewbibmacro{in:}{} 

\DeclareSourcemap{
  \maps[datatype=bibtex]{
    \map[overwrite]{
      \step[fieldsource=doi, final]
      \step[fieldset=url, null]
      \step[fieldset=eprint, null]
    }  
  }
}

\DeclareSourcemap{
  \maps[datatype=bibtex]{
    \map[overwrite]{
      \step[fieldsource=eprint, final]
      \step[fieldset=pages, null]
      \step[fieldset=eid, null]
      \step[fieldset=journal, null]
    }  
  }
}

\bibliography{bibliography}

\newcommand{\qsp}[2]{\,\ensuremath{\raise.5ex\hbox{$#1$}\big\slash\raise-.5ex\hbox{$#2$}}}

\newcommand{\calV}{\mathcal{V}}
\newcommand{\dd}{\mathrm{d}}
\newcommand{\ndash}{\nobreakdash-\hspace{0pt}}
\renewcommand{\Tilde}{\widetilde}

\DeclareMathOperator{\Ima}{Im}

    \makeatletter
        \newcommand{\zzlabel}[1]{\ifmeasuring@\else\ltx@label{#1}\fi} 
    \makeatother

    \newcounter{terms}[equation] 
    \newcommand{\unl}[2]{\underline{#1}_{\refstepcounter{terms} \zzlabel{#2} \theterms}} 

    \newcommand{\reft}[2]{(\ref{#1}.\ref{#2})} 

\title[Boundary BFV action in PC formalism]{Boundary structure of General Relativity in tetrad variables}

\author{G. Canepa}
\address{Institut f\"ur Mathematik, Universit\"at Z\"urich, Winterthurerstrasse 190, 8057 Z\"urich, Switzerland}
\email{giovanni.canepa@math.uzh.ch}

\author{A. S. Cattaneo}
\address{Institut f\"ur Mathematik, Universit\"at Z\"urich, Winterthurerstrasse 190, 8057 Z\"urich, Switzerland}
\email{cattaneo@math.uzh.ch}

\author{M. Schiavina}
\address{Institute for Theoretical Physics, ETH Zurich, Wolfgang Pauli strasse 27, 8093, Z\"urich, Switzerland and Department of Mathematics, ETH Zurich, R\"amistrasse 101, 8092, Z\"urich, Switzerland }
\email{micschia@phys.ethz.ch}

\thanks{This research was (partly) supported by the NCCR SwissMAP, funded by the Swiss National Science Foundation. G.C. and A.S.C. acknowledge partial support of SNF Grant No. 200020- 172498/1. M.S. acknowledges partial support from Swiss National Science Foundation grants P2ZHP2\_164999 and P300P2\_177862.}

\begin{document}

\begin{abstract}

An explicit, geometric description of the first-class constraints and their Poisson brackets for gravity in the Palatini--Cartan formalism
(in space--time dimension greater than three)
 is given. The corresponding  Batalin--Fradkin--Vilkovisky (BFV) formulation is also developed. 
\end{abstract}

\maketitle
\tableofcontents

\section{Introduction}
In this article we clarify the geometry of the boundary structure---in particular, the
reduced phase space---of general relativity in the Palatini--Cartan formalism and develop its Batalin--Fradkin--Vilkovisky (BFV) formulation in any space--time dimension greater than three. 

The Palatini--Cartan (PC) formalism is classically equivalent to the Einstein--Hilbert formalism on a closed manifold, in the sense that they have the same space of solutions of the Euler--Lagrange equations modulo symmetries. However, the PC formalism has several advantages, the main one for us being that employing differential forms allows for a more natural restriction to boundaries, whose study is our main motivation.

The Hamiltonian description of a field theory, which is a particular case of the boundary study when the boundary is a Cauchy surface, has historically been done in terms of Dirac's constraint analysis \cite{Dirac1958}. 
This procedure is rather involved, leading to the analysis of primary and secondary constraints, which then are regrouped into first and second class constraints, to produce eventually the correct space as the symplectic reduction of some submanifold in a symplectic space of fields near the boundary. The final output of this procedure is called \emph{reduced phase space}, and it turns out to be particularly complicated in the case of the Palatini--Cartan formalism.

Typically, one wishes to implement the second class constraints first, and present the reduced phase space as the reduction of a submanifold determined by first-class contraints only.\footnote{This is called a coisotropic submanifold in symplectic geometry.} One advantage of this is that the associated Hamiltonian vector fields can now be interpreted as generators of symmetries, so that the reduced phase space can be regarded as the quotient of a submanifold (fields satisfying the generalised Gauss laws) by gauge transformations. 
In some cases---e.g., Yang--Mills theory and three-dimensional gravity---this can be interpreted as a Marsden--Weinstein \cite{MW1974} reduction.   

A second advantage of obtaining the reduced phase space from first-class constraints only is that it can be cohomologically resolved in terms of the Batalin--Fradkin--Vilkovisky (BFV) formalism \cite{BV1, BV3, Stasheff1997, Schaetz:2008}.
Namely, one considers some appropriate symplectic supermanifold and recasts the constraints into an odd functional (the BFV action) that Poisson commutes with itself. 
This produces a complex whose degree-zero cohomology is isomorphic, as a Poisson algebra, to the algebra of functions of the reduced phase space when the latter is smooth. The main advantage, then, is that one can take this procedure as a definition for the reduced phase space when it is not smooth. Moreover, one can attempt at quantising the reduced phase space in terms of an appropriate quantisation of the supersymplectic manifold (which is often geometrically simpler) and of the BFV action.

The BFV formalism is closely related to the Batalin--Vilkovisky (BV) formalism \cite{BV1, BV2}, which generalises the Faddeev--Popov and BRST constructions, providing a gauge fixing framework for a field theory in the bulk, in view of its perturbative quantisation. The strict connection between BV and BFV, related to a quantisation for manifolds with boundary compatible with cutting and gluing, has been analyzed in \cite{CMR2012, CMR2}. 

In several instances, the BV formalism in the bulk produces a compatible BFV formalism on the boundary \cite{CMR2012}. 
Unfortunately, this is not the case for four-dimensional gravity in the Palatini--Cartan formalism \cite{CS2017}, 
at least in a natural implementation of the BV framework, which otherwise works for the analogous three-dimensional case\footnote{There are examples of theories where one can modify the bulk BV formalism to make it compatible with the boundary BFV formalism \cite{CS2016a}; however this fix currently seems out of reach for PC theory.} \cite{CaSc2019}. 

This paper is a first step in a plan to overcome the problem encountered in \cite{CS2017}: namely, reversing the BV-BFV procedure by first studying the BFV formalism for Palatini--Cartan gravity on the boundary and then inducing a compatible BV formalism in the bulk (this second step is considered in the follow up work \cite{CCS2020b}). 
Another motivation for our study of the BFV structure is to extend the analysis of General Relativity to corners of higher codimension, as was successfully done for other BV-BFV theories \cite{CMR2012,MSW2019, CaSc2019}.

Our solution to the above problem is based on a more geometric alternative to Dirac's construction of the reduced phase space, as introduced by Kijowski and Tulczijew \cite{KT1979}. This alternative has several advantages, simplifying many computations and making them more transparent. Moreover, it also usually produces the reduced phase space as a coisotropic reduction --- i.e., only first class constraints appear --- and, finally, it is closely related to the BV-BFV construction (see \cite{CMR2012b} for the general framework and \cite{CSEH, CS2017} for examples in the context of General Relativity).

In 2017 the last two authors successfully applied this construction to four-dimensional gravity in the Palatini–Cartan formalism \cite{CS2019} for a timelike or spacelike boundary, showing in particular that only first-class constraints appear. Recently, a presentation in terms of first-class constraints only  in the context of Dirac's formulation has been obtained
in \cite{MRC:2019}, with its extension to higher dimension discussed in \cite{MERC:2019}. 

Some of the expressions presented in \cite{CS2019} were not quite as explicit as one might have liked. Although this is does not hinder the theorems on the classical (Hamiltonian) structure, a more explicit description would be desirable when writing down an explicit BFV action for the theory, or for further explicit computations. In this paper we provide such a description, improving the understanding of the reduced phase space of Palatini--Cartan theory and extending all the results to higher dimensions. This allows us to construct the BFV action for PC theory in dimension $N\geq 3$ for a timelike or spacelike boundary. For a lightlike boundary we refer to \cite{CCT2020} where the construction presented here is adapted to the case where the boundary metric is degenerate.
In doing this, we also prove several technical properties of ``tetrads,'' which may be useful also elsewhere. 

\subsection{Structure of the paper}
In Section \ref{s:overview} we summarise the basics of PC theory, review the construction of its reduced phase space (following \cite{CS2019}), and present a new idea that will be used throughout to simplify the boundary structure.

Section \ref{s:technical} is a collection of necessary (technical) results, which expands on the fundamental observation that the map $e^{N-k}\wedge \cdot$ might have a nontrivial kernel, with several consequences.

In Section \ref{s:constraints} we construct the reduced phase space of PC theory using the clever choice presented in Section \ref{s:optimalchoice}: we show that the constraints are first class, and compute their Poisson brackets explicitly.  

Section \ref{s:BFV} is devoted to the construction of the BFV data for Palatini--Cartan theory in dimension $4$, while Section \ref{s:generalisation} generalises all the previous results to $N\geq 5$.

Section \ref{s:RPSext} depends on Section \ref{s:constraints}, but is completely independent of Section \ref{s:BFV}, which is required only by \ref{s:BFVext} and can be ignored by a reader who is interested in purely classical (non-BFV) considerations.

\section{A (short) overview}\label{s:overview}
Our geometric construction is based on \cite{CS2019}, where the last two authors treated the four-dimensional case applying the construction of Kijowski and Tulczijew \cite{KT1979}. In this context, the reduced phase space is obtained as the reduction by first class constraints of an appropriate space of boundary fields. 

The aim of this paper is to supplement the construction of \cite{CS2019} with a more explicit presentation of boundary data.

In this section we will review the generalities of the Palatini--Cartan formalism, summarise the main results of \cite{CS2019}, and present the new idea from which this paper stems.

\subsection{General Relativity in the Palatini--Cartan formulation}
The dynamical field of general relativity in the usual formulation by Einstein and Hilbert is a Lorentzian metric $g$ and the action functional is
\[
S_\text{EH} = \int_M (R-\Lambda)\sqrt g,
\]
where $M$ is space--time, $R$ the scalar curvature of $g$, $\Lambda$ the cosmological constant (a fixed parameter) and $\sqrt g$ the density induced by $g$.

In this paper we focus on the classically equivalent formulation that goes under the name of Palatini--Cartan (or also Palatini--Cartan--Holst in the four\ndash dimensional case \cite{Holst}). It is based on Palatini's calculation of the variation of Riemann's tensor in terms of the Christoffel symbols (known as Palatini identity \cite{Palatini1919}), 
later extended to the idea of treating the connection as an independent field, and on 
Cartan's observation \cite{Cartan} 
that a metric may be alternatively presented in terms of a local frame.

Let $M$ be an $N$-dimensional manifold that admits a Lorentzian structure. In Palatini--Cartan theory one fixes a vector bundle $\calV$ isomorphic to $TM$ and endowed with a fibrewise Minkowski metric $\eta$,\footnote{We consider throughout the paper only the physical case with Lorentzian signature, although our results directly extend to the Euclidean case.} which we also denote by $(\ ,\ )$. Sometimes the vector bundle $\calV$ is referred as the ``fake tangent bundle''. One may, e.g., take $\calV=TM$. In any case we assume that the isomorphism is orientation preserving.

The theory has two dynamical fields.\footnote{The choice of
$\calV$ and $\eta$ is immaterial. Different choices will produce equivalent field theories related by linear redefinitions of the fields.}
 The first is a Cartan coframe, i.e., an orientation preserving bundle isomorphism covering the identity from the tangent bundle to $\calV$ $$e \colon TM \stackrel{\sim}{\longrightarrow}\calV .$$
The coframe field is also known as the \emph{tetrad} (or \emph{vierbein}) in four dimensions.
A Lorentzian metric is recovered (relating to Einstein--Hilbert theory) as
\begin{equation}\label{e:g}
g=e^*\eta; \qquad g_{\mu\nu}=(e_\mu,e_\nu).
\end{equation}
Note that there is more redundancy in $e$ than in $g$. As a consequence, the theory is invariant under additional gauge transformations.

The second dynamical field is an orthogonal connection $\omega$ on $\calV$.
We denote the space of such connections with $\mathcal{A}(M)$. The isomorphism $e$ allows transforming the connection $\omega$ into an affine connection $\Gamma$ that is automatically compatible with the metric $g$---a metric connection.

To write down the action functional, it is useful to introduce a piece of notation: by $\Omega^{i,j}$ we denote the space of sections of $\bigwedge^iT^*M\otimes\bigwedge^j\calV$ ($i$\ndash forms taking values in the $j$th exterior power of $\calV$).
We may then regard the coframe $e$ as an element of $\Omega^{1,1}$ (plus the nondegeneracy condition that it actually defines an orientation preserving isomorphism). Moreover, using the fibre metric $\eta$ one can easily see that the space of orthogonal connections is modelled on $\Omega^{1,2}$ (this is essentially just the fact that the Lie algebra of orthogonal transformations is isomorphic, via the metric, to that of skew-symmetric bilinear forms). In particular, we will regard the curvature $F_\omega$ of $\omega$ as an element of $\Omega^{2,2}$. Furthermore, throughout the article we use the shorthand notation $e^{k}$ to denote $k$th wedge power of $e$ and omit the wedge product symbol in the formulas: the wedge products both in $\bigwedge^\bullet T^*M$ and in $\bigwedge^\bullet\calV$ will be always tacitly understood.

The action functional for Palatini--Cartan theory reads
\[
S = \int_M\left[\frac1{(N-2)!}e^{N-2}F_\omega - \frac1{N!}\Lambda e^N\right].
\]
Note that each term belongs to $\Omega^{N,N}$, which can be canonically identified, via $\sqrt{|\det \eta|}$, with the space of densities on $M$.\footnote{An element of $\Omega^{N,N}$ is a section of $\det T^*M\otimes\det \calV$. On the other hand, $\sqrt{|\det \eta|}$ is a section of $|\det V^*|$, so their product is a section of
$\det T^*M\otimes\mathrm{or}(\calV)$, where $\mathrm{or}(\calV)$ is the orientation bundle of $V$. Under our assumption that the isomorphism between $TM$ and $\calV$ is orientation preserving, we have 
$\mathrm{or}(TM)=\mathrm{or}(\calV)$, so the product of an element of $\Omega^{N,N}$ with $\sqrt{|\det \eta|}$
is a section of $|\det T^*M|$, i.e., a density.}
For ease of notation, we will omit writing down the factor $\sqrt{|\det \eta|}$ explicitly.\footnote{\label{f:etaconst}It is actually possible to choose $\calV$ in such a way that $\sqrt{|\det \eta|}$ is equal to one. Namely, pick a Lorentzian metric on $M$ and reduce its frame bundle to the orthogonal frame bundle $P$. Then one can define $\calV$ as the associate bundle $P\times_{O(N-1,1)}W$, where $W$ is the fundamental representation, endowed with the Minkowski metric.
With this choice $\eta$ is the constant Minkowski metric, and the transitions function of $\det \calV$ are locally constant and equal to $\pm1$. Moreover, $\det T^*N\otimes\det \calV$ is directly equal to $|\det T^*M|$, so that elements of $\Omega^{N,N}$ are canonically the same as densities.}

\begin{remark}
Note that it is possible to consider other terms in the action, namely
$e^{N-2k}F_\omega^k$ for every $k\le N/2$. These other terms will however yield Euler--Lagrange equations involving higher derivatives of the fields, apart from the term $F_\omega^{N/2}$, which is topological (it is the Holst term in four dimensions). We will not consider these extensions in this paper.
\end{remark} 

The Euler--Lagrange equation obtained by a variation of $\omega$ is $\dd_\omega(e^{N-2})=0$, where
$\dd_\omega$ denotes the covariant derivative $\Omega^{\bullet,\bullet}\to\Omega^{\bullet+1,\bullet}$ associated to $\omega$.\footnote{One gets $\dd_\omega(e^{N-2})=0$ directly with the choice of constant fibre metric as in
footnote~\ref{f:etaconst}.
In general, the Euler--Lagrange equation is $\dd_\omega(\sqrt{|\det \eta|}e^{N-2})=0$, but, since $\omega$ is an orthogonal connection, we have $\dd_\omega\eta=0$ and, therefore, $\dd_\omega\sqrt{|\det \eta|}=0$.
By the Leibniz rule, we may omit the
nonzero factor $\sqrt{|\det \eta|}$.}
By the Leibniz rule this equation may be rewritten as $e^{N-3}\dd_\omega e= 0$, which, by the nondegeneracy condition on $e$,\footnote{The nondegeneracy condition is obviously not necessary in case $N=3$.} is equivalent to
\begin{equation}\label{e:tf}
\dd_\omega e = 0.
\end{equation}
It may be easily shown that this condition is equivalent to the condition that the affine connection $\Gamma$ induced by $\omega$ be torsion free. Since $\Gamma$ is also metric, it must then be the Levi-Civita connection, and this determines a unique $\omega_e$ solving \eqref{e:tf} for a given $e$.

The Euler--Lagrange equation obtained by a variation of $e$ is  
\begin{equation}\label{e:ee}
\frac1{(N-3)!}e^{N-3}F_\omega-\frac1{(N-1)!}\Lambda e^{N-1}=0.
\end{equation}
Inserting $\omega_e$, this equation turns out to be equivalent to Einstein's equation for the metric $g$ defined in
\eqref{e:g}.

\begin{remark}
Truly, to obtain Equations \eqref{e:tf} one needs injectivity of the map $e^{N-3}\wedge$. From the results of Lemma \ref{lem:We_bulk} will show that $e^{N-3}\wedge d_{\omega}e=0$ is indeed equivalent to \eqref{e:tf}, while no further simplifications can be applied to \eqref{e:ee}. Moreover, this observation will turn out to be true only in the bulk, and will play a crucial role in the definition of boundary variables and constraints (see Sections \ref{s:technical} and \ref{ssec:constraints}).
\end{remark}

Although the solution is proposed for a generic dimension $N\geq 4$, the remainder of this overview will focus on the four-dimensional case:
$N=4$. In this case the Palatini--Cartan action functional is simply
\begin{equation}\label{e:simplePCaction}
S = \int_M \frac12 eeF_\omega + \frac{1}{4!}\Lambda e^4
\end{equation}
and its
Euler--Lagrange equations are (equivalent to)
\begin{equation}\label{e:EL}
\dd_\omega e = 0,\qquad eF_\omega + \frac{1}{3!}\Lambda e^3=0.
\end{equation}

\subsection{Reduced phase space for Palatini--Cartan theory} \label{sec:RPS-PC_intro}

We apply to this theory the construction due to Kijowski and Tulczijew \cite{KT1979}, which allows to investigate its phase space and the Hamiltonian formulation.
\begin{remark}
In order to keep the notation simple, we will denote throughout the paper the boundary of a manifold $M$ with $\Sigma:= \partial M$.
\end{remark}

We begin by observing that, when varying the action \eqref{e:simplePCaction}, one gets a boundary term\footnote{This is precisely so with the choice of constant fibre metric as in
footnote~\ref{f:etaconst}. In general, in this formula, as well as in the formulas for the constraints that will appear later, there is a hidden factor $\sqrt{|\det \eta|}$---which has no effects on the computations and results--- but is required to produce densities, which can be canonically integrated. {}From now on we will no longer mention this factor.}
\[
\Tilde\alpha^\partial = \frac12\int_{\Sigma}  ee\delta\omega.
\]
This is the analogue of the $p\dd q$ term in classical mechanics.  We view the restrictions of $e$ and $\omega$ to the boundary as, respectively,
a nondegenerate section of $T^*(\Sigma)\otimes\calV|_{\Sigma}$ --- i.e., as an injective bundle map $T(\Sigma)\to\calV|_{\Sigma}$ --- and an orthogonal connection associated to $\calV|_{\Sigma}$. Again, we may view the space of these connections as modeled on $T^*(\Sigma)\otimes\bigwedge^2\calV|_{\Sigma}$.

We can then regard $\Tilde\alpha^\partial$ as a one-form on the space $\widetilde{F}_{PC}$ of the \emph{pre-boundary} fields $e\vert_{\Sigma}$ and $\omega\vert_{\Sigma}$. Thus, we might think of $\Tilde\varpi=\delta\Tilde\alpha^\partial$ as a ``pre-symplectic form'' on the space of pre-boundary fields. In fact, the two-form $\Tilde\varpi$ is degenerate: a vector field $X$ in the kernel of $\Tilde\varpi$ acts as $\omega\to\omega + v$ with 
\begin{equation}\label{e:Xomega}
e v=0.
\end{equation} 

\begin{remark}\label{rem:presympred}
We stress that this transformation implicitly depends on $e$ and, under the nondegeneracy assumption on $e$, the $v$s satisfying \eqref{e:Xomega} and hence the $X$s in the kernel of $\Tilde\varpi$ have exactly 6 local components. If we mod out the space of pre-boundary fields by the kernel of $\Tilde\varpi$ we get a space parametrised by
$e$ and equivalence classes of $\omega$ under the $e$-dependent transformation above.\footnote{Observe that $e$ and the remaining $\omega$ both have $12$ local components, or \emph{degrees of freedom}.} This defines the map
\begin{equation}\label{e:presymplecticreduction}
\pi_{PC}\colon \widetilde{F}_{PC} \longrightarrow {F}_{PC}.
\end{equation}
On this quotient space, the two-form  $\Tilde\varpi$ determines a nondegenerate, closed two-form: the manifold $({F}_{PC},  \varpi_{PC})$ is the \emph{geometric} phase space of the theory.
\end{remark}

The symplectic manifold defined by \eqref{e:presymplecticreduction} is not yet the ``physical'' phase space of the theory, usually called \emph{reduced phase space}. Indeed, the Euler--Lagrange equations \eqref{e:EL} split into evolution equations, which contain derivatives of the pre-boundary fields in a transversal direction, and equations where only tangential derivatives appear. The latter equations, called the constraints, must be imposed on the preboundary fields, but this enlarges
the kernel of the presymplectic form, and the corresponding reduction has to be taken into account. To obtain the reduced phase space, it is advantageous to reformulate this procedure in terms of the geometric phase space we have introduced above. 

\begin{remark}\label{rem:naiveconstraints}
An advantage of the Palatini--Cartan formulation is that it is formulated in terms of differential forms and, as a consequence, the constraints are readily available as the restriction to the boundary of Equations \eqref{e:EL} \footnote{Note that the relevant quantity is the zero locus generated by the constraints and not the actual functional form of them.}. One problem is that the constraints are not necessarily invariant under the transformations generated by $X$ in the kernel of $\Tilde\varpi$, i.e. translations of $\omega$ by $v$ (and in fact they are not). There are two possible ways out: to select the $v$\ndash invariant parts of the constraints and take the quotient by the $v$'s 
or to look for a section to the $v$\ndash translations. As in \cite{CS2019},
we will follow here the second strategy.
\end{remark}

The first remark is that the constraint $eF_\omega+ \frac{1}{3!}\Lambda e^3=0$ is indeed $v$\ndash invariant upon using the first constraint $\dd_\omega e = 0$.\footnote{If we denote by $\delta_v$ a variation along $v$, we get
\[
\delta_v (eF_\omega)= e\dd_\omega v = \dd_\omega(ev)-\dd_\omega e\,v. 
\]
The first term vanishes because $ev=0$ and the second because we assume $\dd_\omega e = 0$.
}
Therefore, it is better to use a $v$\ndash section that, in conjunction with the invariant part of $\dd_\omega e = 0$,
reproduces the whole constraint. 

It is easy to check that the induced constraint 
\begin{equation}\label{e:Invariantconstraint}
e\dd_\omega e = 0
\end{equation}
($6$ local components) is indeed $v$\ndash invariant. We will call it the \emph{invariant constraint}. It turns out that Equation \eqref{e:Invariantconstraint} determines the whole invariant part--- under the condition that $e$ is such that the boundary metric
\begin{equation} \label{e:Boundary-metric}
g^\partial_{ij}:=(e_i,e_j)
\end{equation}
 is nondegenerate,\footnote{This is for example the case when $M=B\times [0,1]$, where $[0,1]$ is an interval, and $e$ is assumed to produce a metric for which $B\times\{0\}$ and $N\times\{1\}$ are space-like.} where $i,j$ are indices of boundary coordinates. {}From now on we will assume this condition. 

Since the remaining components of the constraint $\dd_\omega e = 0$, which we call the \emph{structural constraint}, are also 6, they can now be used to fix the $v$-translations completely. Note that the invariant constraints are canonically given, whereas the structural ones require a choice. 

A few remarks are now in order (see \cite{CS2019} for their proofs): 
\begin{enumerate}
\item Since the  structural constraint completely fixes the $v$\ndash transformations, the space $\mathsf{S}$ of pre-boundary fields satisfying it is symplectomorphic to the space of boundary fields. 
\item On $\mathsf{S}$ ($12$ local degrees of freedom) we still have to impose the $10$ local constraints\footnote{These constraints look like the restriction to the boundary of the Euler--Lagrange equations, except $e\wedge$ cannot be eliminated from either expression. Note also that now part of the components of $\omega$ are constrained, so these constraints are only formally equal to the restriction of the Euler--Lagrange equations.}
\[
e\dd_\omega e = 0,\qquad eF_\omega + \frac{1}{3!}\Lambda e^3=0.
\]
These constraints are first class \cite{CS2019}, and we are then left with the expected $2$ local physical degrees of freedom of four-dimensonal gravity.
\item The constraints may be written in terms of Lagrange multipliers $c$ and $\mu$ as
\begin{equation}\label{e:JK}
L_c=\int_\Sigma ce\dd_\omega e,\qquad J_\mu=\int_\Sigma\mu \left(eF_\omega + \frac{1}{3!}\Lambda e^3=0\right).
\end{equation}
It turns out that ``on shell,'' i.e., upon the constraints, $L$ generates the internal gauge transformations and $J$ the diffeomorphisms (including the remnant of the transversal ones).
\item One may also reduce by stages. One possibility is to impose $e\dd_\omega e = 0$ and to mod out by gauge transformations. The resulting space, with $6$ local degrees of freedom, is symplectomorphic to the phase space of the Einstein--Hilbert formulation (the ``cotangent bundle'' of the space of boundary metrics). The remaing constraints
$eF_\omega+ \frac{1}{3!}\Lambda e^3=0$ produce the energy and momentum constraints. Another possibility is to split the Lie algebra of orthogonal transformations into two $3$-dimensional subalgebras. The symplectic reduction with respect to one of the summands yields Ashtekar's formulation \cite{Ashtekar1986}.
\end{enumerate}

\begin{remark}
Note that the Hamiltonian vector fields of $L$ and $J$ in \eqref{e:JK} depend on the actual choice of the structural constraints and may be not very explicit if the choice is not optimal. 
This is not a serious problem for the classical considerations above, and for this reason no attempt to find an optimal choice was made in \cite{CS2019}; 
however, a non optimal choice is inconvenient for concrete computations as well as
for further considerations like, e.g., the explicit BFV description of the theory. 
\end{remark}

\subsection{An optimal choice of structural constraints}\label{s:optimalchoice}
The main result in this paper is to present a choice of structural constraints and how to use it to produce the BFV data associated to the reduced phase space of Palatini--Cartan theory (see Section \ref{s:BFV}). A related explicit choice, in the context of Dirac's formulation, has been presented in \cite{MRC:2019}, with its extension to higher dimension discussed in \cite{MERC:2019} (unrelated to BFV).

First of all we choose a section $e_n$ of $\calV|_\Sigma$ that is a completion of the basis $e_1,e_2,e_3$ (here $1,2,3$ denote indices of boundary coordinates). Note that in a neighborhood of a given $e$ in the space of pre-boundary fields we may choose $e_n$ once and for all independently of the $e$'s in the neighborhood. This done, we write the structural constraints as
\begin{equation}\label{e:structuralconstraint}
e_n\dd_\omega e = e\sigma
\end{equation}
for some unspecified one-form $\sigma$ taking values in $\calV|_\Sigma$. Note that we have $18$ equations with
$12$ unspecified parameters $\sigma$, so in total we have indeed $6$ constraints. 

We will show that this choice of structural constraint fixes the $v$\ndash translations and that, together with the invariant constraint $e\dd_\omega e=0$, it produces the full constraint $\dd_\omega e=0$, which is necessary for the $v$\ndash invariance of $eF_\omega + \frac{1}{3!}\Lambda e^3=0$. Moreover, we will show that this choice actually makes the Hamiltonian vector fields of  $L$ and $J$ in \eqref{e:JK} explicit enough to allow writing down the BFV action of the theory. Finally, it will allow us to extend the result in the presence of a cosmological term and to higher dimensions.

\begin{remark}
Observe that, although not necessary, one may interpret the linearly independent system $(e_1,e_2,e_3,e_n )$ as a coframe in a neighborhood of $\Sigma$ in $M$ and the structural constraint
as one of the remaining Euler--Lagrange equations, with $\sigma$ interpreted as the transversal components of
$\dd_\omega e$. Viewed this way, the structural constraint \eqref{e:structuralconstraint} also immediately shows that the transversal Euler--Lagrange equations (the evolution equations) may actually be solved.
\end{remark}

\section{Technical results}\label{s:technical}
In this section we collect  some technical lemmas that will be useful throughout the paper. We postpone the proofs of lemmas \ref{lem:We_bulk}, \ref{lem:We_boundary} and \ref{lem:varrho12} to Appendix \ref{sec:appendix-CME}. Let us fix the notation. From now on we will use the notation $\calV$ also for its restriction to the boundary. 
For $\mathrm{dim}(M)=N=\mathrm{dim}(\mathcal{V}_x)$, on
$$
\Omega^{i,j}:= \Omega^i\left(M, \textstyle{\bigwedge^j} \mathcal{V}\right) \qquad \Omega_{\partial}^{i,j}:= \Omega^i\left(\Sigma, \textstyle{\bigwedge^j} \mathcal{V}\right)$$ 
 
we define the linear maps:

\begin{eqnarray}
W_{k}^{ (i,j)}: \Omega^{i,j}  \longrightarrow &\Omega^{i+k,j+k}  \label{W_e-bulk}\\
X  \longmapsto  & X \wedge \underbrace{e \wedge \dots \wedge e}_{k-times}  \nonumber\\
W_{k}^{ \partial, (i,j)}: \Omega_{\partial}^{i,j}  \longrightarrow & \Omega_{\partial}^{i+k,j+k} \label{W_e-boundary} \\
X  \longmapsto & X \wedge \underbrace{e \wedge \dots \wedge e}_{k-times}.\nonumber
\end{eqnarray}
The properties of these maps will be clarified by the following results. They will turn out to be crucial in shaping the boundary structure of Palatini--Cartan theory.
We will consider elements in $\Omega^{i,j}$ and $\Omega_{\partial}^{i,j}$ to have total degree $i+j$ and the wedge product, which we implicitly use in the formulas below, defines a graded commutative associative algebra $\Omega^{\bullet,\bullet}$ with respect to the total degree.\footnote{For $\alpha \in \Omega^{i,j}$ and $\beta \in \Omega^{k,l}$ we have $\alpha \wedge \beta = (-1)^{(i+j)(k+l)}\beta \wedge \alpha$. In particular $e$ is an even element of $\Omega^{\bullet,\bullet}$.}
\begin{lemma} \label{lem:We_bulk}
Let $N=\mathrm{dim}(M)\geq 4$. Then 
\begin{enumerate}
\item $ W_{N-3}^{ (2,1)}$ is bijective;  \label{lem:We21}
\item $ \mathrm{dim}\mathrm{Ker}W_{N-3}^{(2,2)}\not=0$. \label{lem:We22}
\end{enumerate}
\end{lemma}

\begin{lemma} \label{lem:We_boundary}
The maps $W_{k}^{ \partial, (i,j)}$ have the following properties for $N \geq 4$:

\begin{enumerate}
\item $W_{N-3}^{\partial, (2,1)}$ is surjective; \label{lem:Wep21}
\item $W_{N-3}^{\partial, (1,1)}$ is injective; \label{lem:Wep11}
\item $W_{N-3}^{\partial, (1,2)}$ is surjective; \label{lem:kerWe12} 
\item $\dim \mathrm{Ker} W_{N-3}^{\partial, (1,2)} = \dim \mathrm{Ker} W_{N-3}^{\partial, (2,1)}$;\label{lem:kernel12-21}
\item $W_{N-4}^{\partial, (2,1)}$ is injective. ($N \geq 5$)\label{lem:We5p21}
\end{enumerate}
\end{lemma}

We can also define a map
\begin{align*}
\varrho : \Omega_{\partial}^{1,2}  & \longrightarrow \Omega_{\partial}^{2,1} \\
X & \longmapsto [X, e] .
\end{align*}
It has the following property:
\begin{lemma}\label{lem:varrho12}
If $g^\partial$, as defined in \eqref{e:Boundary-metric}, is nondegenerate, then
$\varrho |_{\mathrm{Ker} W_{N-3}^{\partial, (1,2)}}$ is injective.
\end{lemma}

\begin{remark}
Some of the properties in Lemmas \ref{lem:We_bulk}, \ref{lem:We_boundary} and \ref{lem:varrho12} have already been proven in \cite{CS2019} for $N =4$. In Appendix \ref{sec:appendix-CME} we will follow a similar strategy for their proofs, adapting them to the different dimensions. In \cite[Lemma 4.12]{CS2019}, a map similar to $\varrho$ was used, denoted by $\phi_e$, which is the restriction of $\varrho$ to the kernel $\mathrm{Ker}W^{\partial, (1,2)}_{1}$, composed with the projection $p_{2,1}$ to $\mathrm{Ker}W^{\partial, (2,1)}_1$, that is to say $\phi_e\equiv p_{(2,1)}\circ \varrho\vert_{\mathrm{ker} W^{\partial, (1,2)}_1}$. 
\end{remark}
\begin{remark}
Throughout the paper we will refer to the \emph{dimensions} (as $C^\infty$ modules) of the spaces $\Omega^{i,j}$ as the number of degrees of freedom of the space. Note that this dimension is also the same as the rank of the typical fibre. Hence for example $\dim (\Omega^{i,j}) := \dim \bigwedge^{i} (T_x^*M) \times\bigwedge^{j}\mathcal{V}_x= \binom{N}{i}\binom{N}{j}$. 
\end{remark}

Recalling the definition of  $e_n$ in Section \ref{s:optimalchoice} as a section of $\calV|_\Sigma$ that is a completion of the basis $e_1,e_2,e_3$, we can state the following:
\begin{lemma}\label{lem:Omega2,1_d4}
Let $\alpha \in \Omega^{2,1}_\partial$. Then 
\begin{align}\label{ConditionforOmega21_d4}
\alpha=0 \qquad \Longleftrightarrow  \qquad \begin{cases}
e^{N-3}\alpha =0 \\
e_n e^{N-4}\alpha \in \Ima W_{N-3}^{\partial, (1,1)}
\end{cases}.
\end{align}
\end{lemma}
\begin{proof}
We first note that the second requirement corresponds to the existence of a $\sigma \in \Omega_{\partial}^{1,1}$ such that $ e_n e^{N-4} \alpha = e^{N-3} \sigma$.
Let now $I \subset \mathbb{R}$ be an interval and let $x^n$ be the coordinate along it. We define $\widetilde{M} = \Sigma \times I$ and rewrite \eqref{ConditionforOmega21_d4} as conditions on the pullbacks of $e, e_n, \sigma$ and $\alpha$ to $\widetilde{M}$, which we will keep denoting with the same letters.  We now define the following forms on $\widetilde{M}$:
$$ E = e^{N-3} + e_n e^{N-4}d x^n , \qquad A = \alpha + \sigma dx^n.$$
Hence the system \eqref{ConditionforOmega21_d4} corresponds to the single equation $ E \wedge A =0$. Since $e_n$ has been chosen to be linearly independent from $e$ as vectors in $\mathcal{V}$, $E$ is an isomorphism $T\widetilde{M} \rightarrow \mathcal{V}$. Hence we can use Lemma \hyperref[lem:We21]{\ref*{lem:We_bulk}.(\ref*{lem:We21})} and deduce that 
$$E \wedge \cdot : \Omega^2(\widetilde{M}, \mathcal{V}) \rightarrow \Omega^{N-1}(\widetilde{M}, \textstyle{\bigwedge^{N-2}\mathcal{V}})$$
is injective. Hence $A=0$, which in turn implies $\alpha=0$.
\end{proof}

\begin{corollary} \label{c:dynamical+structural}
\begin{align*}
d_{\omega}e=0 \qquad \Longleftrightarrow  \qquad \begin{cases}
e^{N-3}d_{\omega}e =0 \\
e_n e^{N-4}d_{\omega}e \in \Ima W_{N-3}^{\partial, (1,1)}
\end{cases}.
\end{align*}
\end{corollary}
\begin{proof}
Trivial application of Lemma \ref{lem:Omega2,1_d4} to $\alpha= d_{\omega}e$.
\end{proof}

\begin{corollary}\label{cor:enve_injective}
If $g^\partial$ is nondegenerate, the map
\begin{align*}
\chi \colon \mathrm{Ker} W_{N-3}^{\partial, (1,2)} & \rightarrow \Omega_{\partial}^{(N-2,N-2)} \\
 v & \mapsto e_n  e^{N-4}[v,e]
\end{align*}
is injective and in particular
\begin{equation}\label{envenotinIme_d4}
\Ima \chi \cap \Ima W_{N-3}^{\partial, (1,1)}= \{0\}.
\end{equation}
\end{corollary}
\begin{proof}
Consider $ 0 \neq v  \in \mathrm{Ker} W_{N-3}^{\partial, (1,2)}$, i.e. such that $e^{N-3}v=0$. We get
\begin{align*}
e^{N-3} [v, e] = [ e^{N-3} v, e] - v [e^{N-3} , e]= (N-3) v e^{N-4}[e,e] = 0.
\end{align*}
 Suppose now by contradiction that $e_n  e^{N-4}[v,e] \in \Ima W_{N-3}^{\partial, (1,1)}$; then, applying lemma \ref{lem:Omega2,1_d4} to $\alpha = [v,e]$, we get $[v,e]=0$. From Lemma \ref{lem:varrho12} we know that if $g^\partial$ is nondegenerate, $[v,e] \neq 0$ which contradicts the previous assertion. 
\end{proof}

\begin{lemma}\label{lem:Omega2,2_d4}
Let $\beta \in \Omega^{N-2,N-2}_\partial$. If $g^\partial$ is nondegenerate, there exist a unique $v \in \mathrm{Ker} W_{N-3}^{\partial, (1,2)}$ and a unique $\gamma \in \Omega_{\partial}^{1,1}$ such that 
\begin{align*}
\beta = e^{N-3} \gamma + e_n e^{N-4} [v, e].
\end{align*}
\end{lemma}
\begin{proof}
From Lemma \hyperref[lem:Wep11]{\ref*{lem:We_boundary}.(\ref*{lem:Wep11})} and Lemma \hyperref[lem:kernel12-21]{\ref*{lem:We_boundary}.(\ref*{lem:kernel12-21})} we know that $W_{N-3}^{\partial, (1,1)}$ is injective and that the sum of the dimensions of $\mathrm{Ker} W_{N-3}^{\partial, (1,2)}$ and of $\Ima W_{N-3}^{\partial, (1,1)}$ agrees with dimension of $\Omega^{N-2,N-2}_\partial$. Using Corollary \ref{cor:enve_injective}, we deduce that $\Omega^{N-2,N-2}_\partial$ is the direct sum of $\Ima \chi $ and $ \Ima W_{N-3}^{\partial, (1,1)}$.
Hence every $\beta \in \Omega^{N-2,N-2}_\partial$ can be written as $\beta =e^{N-3}  \gamma + \theta $ with $\gamma \in \Omega^{1,1}_\partial$ and $\theta = e_n e^{N-4}[v, e]$. Uniqueness of $v$ and $\gamma$ follows from the injectivity of $\chi$ and $W_{N-3}^{\partial, (1,1)}$.
\end{proof}

\section{Constraint analysis of Palatini--Cartan theory in four dimensions} \label{s:constraints}
In this section we analyse the structural and invariant constraints of gravity in the Palatini--Cartan formulation for $N=4$, as discussed in Section \ref{sec:RPS-PC_intro}. In Section \ref{s:generalisation} we will extend this analysis to $N > 4$. We will assume henceforth that 
$g^{\partial}$, as defined in \eqref{e:Boundary-metric},
is nondegenerate. The degenerate case will be analysed in \cite{CCT2020}.

\subsection{An optimal structural constraint} \label{sec:space_of_classical_boundary_fields}
The starting point of our analysis is the geometric phase space ${F}^{\partial}_{PC}$, described in full detail in \cite{CS2019} and recalled in Section \ref{sec:RPS-PC_intro}.
The classical fields of the theory are then $ e \in \Omega_{nd}^1(\Sigma, \mathcal{V})$ --- i.e $ \Omega_{\partial}^{1,1}$ plus the nondegeneracy condition that the induced morphism $T\Sigma\to\calV$ should be injective--- and the equivalence class of a connection $\omega \in \mathcal{A}(\Sigma)$ (where $\mathcal{A}(\Sigma)$ is the restriction of $\mathcal{A}(M)$ to the boundary) under the $e$-dependent relation $\omega  \sim \omega + v $ for $v$ such that $e \wedge v =0$. We denote this equivalence class and the quotient space it belongs to by $[\omega] \in \mathcal{A}^{red}(\Sigma)$. The symplectic structure is given by
\begin{align}\label{classical-boundary-symplform}
\varpi = \int_{\Sigma} e \delta e \delta [\omega].
\end{align}
In this section we fix a convenient representative for this equivalence class. 

As in Section \ref{s:optimalchoice} we choose a section of $\calV|_\Sigma$ completing the image of $e\colon T\Sigma \rightarrow \mathcal{V}$ to a basis.
Corollary \ref{c:dynamical+structural} shows that the constraint $d_\omega e = 0$ splits into the \emph{invariant constraint} $e d_\omega e = 0$ and the constraint
\begin{equation}\label{omegareprfix}
e_n d_{\omega} e \in \Ima W_1^{\partial,(1,1)},
\end{equation}
which can then be taken as a choice of \emph{structural constraint}.
We prove that Equation \eqref{omegareprfix} does not impose any condition on $[\omega] \in \mathcal{A}^{red}(\Sigma)$ --- but it fixes a unique representative of the class: In particular we show that given $[\omega]$ there exists a unique $\omega \in [\omega]$ satisfying \eqref{omegareprfix}. Later on we will use such representative to define the constraint of the theory.

\begin{theorem}\label{thm:omegadecomposition}
Suppose that  $g^{\partial}$, the metric induced on the boundary, is nondegenerate. Given any $\widetilde{\omega} \in \Omega^{1,2}$, there is a unique decomposition 
\begin{equation} \label{omegadecomp}
\widetilde{\omega}= \omega +v
\end{equation}
with $\omega$ and $v$ satisfying 
\begin{equation}\label{omegareprfix2}
ev=0 \quad \text{ and } \quad e_n d_{\omega} e \in \Ima W_1^{\partial,(1,1)}.
\end{equation}
\end{theorem}
\begin{proof}
Let $\widetilde{\omega} \in \Omega_\partial^{1,2}$. From Lemma \ref{lem:Omega2,2_d4} we deduce that there exist unique $\sigma \in \Omega_\partial^{1,1}$ and $v \in \text{Ker} W_1^{\partial,(1,2)}$ such that 
\begin{align*}
e_n d_{\widetilde{\omega}} e = e \sigma + e_n [v,e].
\end{align*}
We define $\omega := \widetilde{\omega} - v $. Then $\omega$ and $v$ satisfy \eqref{omegadecomp} and \eqref{omegareprfix2}.

For uniqueness, suppose that $\widetilde{\omega}= \omega_1 + v_1 = \omega_2 +v_2$ with $ev_i =0$ and $e_n d_{\omega_i} e \in \Ima W_1^{\partial,(1,1)}$ for $i=1,2$. Hence 
$$e_n  d_{\omega_1} e- e_n  d_{\omega_2} e = e_n  [v_2-v_1, e] \in  \Ima W_1^{\partial,(1,1)}.$$ Hence from Lemma \ref{lem:Omega2,1_d4} and Lemma \ref{lem:Omega2,2_d4} (for which we need nondegeneracy of $g^\partial$), we deduce $v_2-v_1 =0$, since $v_2-v_1 \in Ker W_1^{\partial,(1,2)}$.
\end{proof}

\begin{remark}
A decomposition similar to \eqref{omegadecomp} was used in \cite[Remark 4.7]{CS2019}, for a generic complement of $\mathrm{Ker}W_1^{\partial,(1,2)}$. Theorem \ref{thm:omegadecomposition} shows an explicit choice of a complement which will turn out to be particularly convenient in what follows.
\end{remark}

\begin{corollary}
The field $\omega$ in the decomposition \eqref{omegadecomp} depends only on the equivalence class $[\omega] \in \mathcal{A}^{red}(\Sigma)$.

\end{corollary}
\begin{proof}
Let $\widetilde{\omega}_1, \widetilde{\omega}_2 \in [\omega]$. Hence $\widetilde{\omega}_1- \widetilde{\omega}_2= \tilde{v} \in \text{Ker} W_1^{\partial,(1,2)}$. Applying Theorem \ref{thm:omegadecomposition} we get $\omega_1$, $v_1$, $\omega_2$, $v_2$ such that $v_1, v_2 \in \text{Ker} W_1^{\partial,(1,2)}$ and 
\begin{align*}
\widetilde{\omega}_1= \omega_1 + v_1 \qquad & e_n d_{\omega_1} e \in \Ima W_1^{\partial,(1,1)} \\
\widetilde{\omega}_2= \omega_2 + v_2 \qquad & e_n d_{\omega_2} e \in \Ima W_1^{\partial,(1,1)}.
\end{align*}
Subtracting these equations we get $\omega_2- \omega_1 = v_1-v_2 - \tilde{v} \in \text{Ker} W_1^{\partial,(1,2)}$ together with
$e_n [\omega_1- \omega_2,e ] \in \Ima W_1^{\partial,(1,1)}$.  Hence, from Lemma \ref{lem:Omega2,2_d4}, we deduce $\omega_1 = \omega_2$.
\end{proof}

\subsection{Poisson brackets of constraints} \label{ssec:constraints}
The restriction of the Euler--Lagrange equations to the boundary does not produce a well defined set of constraints in the geometric phase space ${F}_{PC}$, as they are not given by basic functions with respect to the pre-symplectic reduction $\pi_{PC}\colon \widetilde{F}_{PC} \longrightarrow {F}_{PC}$ (see Remarks \ref{rem:presympred} and \ref{rem:naiveconstraints}).

However, fixing a representative of the equivalence class of $\omega$ by imposing the structural constraint \eqref{omegareprfix} in $\widetilde{F}_{PC}$ (thus constructing a section of the map $\pi_{PC}$) allows us to consider the restrictions of the Euler--Lagrange equations to the boundary and to construct a set of constraints on the geometric phase space. Moreover, we will see that these constraints turn out to be of first class
(i.e. they define a coisotropic submanifold with respect to the symplectic form \eqref{classical-boundary-symplform}). 

Starting from the functions defined in \eqref{e:JK}, we consider the following functions by splitting $R_{\mu}$ into two separate constraints $P_{\xi}$ and $H_{\lambda}$ by expanding $\mu = \iota_{\xi} e + \lambda e_n$. Notice that with this choice of $\mu$ the cosmological term will appear only in the constraint with $\lambda$, since $\iota_{\xi}e^4=0$ on the boundary.  We furthermore add to $P_{\xi}$ a term proportional to the invariant constraint $ed_\omega e$ with the help of a reference connection $\omega_0$ in order to simplify computations (see Remark \ref{r:Kxi}):\footnote{These constraints are a slightly modified but equivalent version of those proposed in \cite{CS2019}, defined on the geometric phase space using the $\omega \in [\omega]$ defined in Theorem \ref{thm:omegadecomposition}, hence satisfying \eqref{omegareprfix}.}
\begin{subequations}\label{constraints}
\begin{equation}
L_c = \int_{\Sigma} c e d_{\omega} e 
\end{equation}
\begin{equation}
P_{\xi}= \int_{\Sigma}  \iota_{\xi} e e F_{\omega} + \iota_{\xi} (\omega-\omega_0) e d_{\omega} e
\end{equation}
\begin{equation}
H_{\lambda} = \int_{\Sigma} \lambda e_n \left(eF_\omega + \frac{1}{3!}\Lambda e^3\right)
\end{equation}
\end{subequations}
where $c \in\Omega^{0,2}_\partial[1]$, $\xi \in\mathfrak{X}[1](\Sigma)$ and $\lambda\in \Omega^{0,0}_\partial[1]$ are (odd) Lagrange multipliers and the notation $[1]$ denotes that the fields are shifted by 1 and are treated as odd variables. 
\begin{remark}
We use odd Lagrange multipliers $c$, $\xi$ and $\lambda$ and we shift their degree by one, to be consistent with the subsequent construction of the BFV action, where we embed our space of fields into a graded manifold, but also in order to simplify the proof of Theorem \ref{thm:first-class-constraints} slightly. However, 
one could just as well formulate constraints \eqref{constraints} using even Lagrange multipliers, and the results of the following Theorem \ref{thm:first-class-constraints} would not change, upon antisymmetrisation of brackets: $\{L_c, L_{c'}\} = \mathbb{L}_c(L_{c'}) - \mathbb{L}_{c'}(L_{c})$, where $\mathbb{L}$ denotes the Hamiltonian vector field of $L_c$ (see \cite{CS2019} for comparison), and similarly for the other constraints.
\end{remark}

\begin{remark} \label{r:Kxi}
The second term in $P_\xi$ does not change the constrained set but largely simplifies the computation of the Hamiltonian vector fields and, consequently, of the Poisson brackets. Indeed, one could just consider $P_{\xi}= \int_{\Sigma}  \iota_{\xi} e e F_{\omega}$ and perform a similar analysis to the one presented in \cite{CS2019}, where the variation $\delta \omega$ is subject to some constraint. Indeed, in section \ref{s:altvar} we will show how to build a covariant expression for the BFV action \eqref{action_C3}, which does not require the choice of a reference connection $\omega_0$. 
\end{remark}

We denote with $\mathrm{L}_{\xi}^{\omega}$ the covariant Lie derivative along the odd vector field $\xi$ with respect to a connection $\omega$:
\begin{align*}
\mathrm{L}_{\xi}^{\omega} A = \iota_{\xi} d_{\omega} A -  d_{\omega} \iota_{\xi} A \qquad A \in \Omega^{i,j}_{\partial}.
\end{align*}

\begin{theorem} \label{thm:first-class-constraints}
 Let $g^\partial$ be nondegenerate on $\Sigma$. Then, the functions $L_c$, $P_{\xi}$, $H_{\lambda}$ are well defined on ${F}^{\partial}_{PC}$ and define a coisotropic submanifold  with respect to the symplectic structure $\varpi_{PC}$. In particular they satisfy the following relations
\begin{subequations}\label{brackets-of-constraints}
\begin{eqnarray}
\{L_c, L_c\} = - \frac{1}{2} L_{[c,c]} & \{P_{\xi}, P_{\xi}\}  =  \frac{1}{2}P_{[\xi, \xi]}- \frac{1}{2}L_{\iota_{\xi}\iota_{\xi}F_{\omega_0}} \\
\{L_c, P_{\xi}\}  =  L_{\mathrm{L}_{\xi}^{\omega_0}c} & \{L_c,  H_{\lambda}\}  = - P_{X^{(a)}} + L_{X^{(a)}(\omega - \omega_0)_a} - H_{X^{(n)}} \\
\{H_{\lambda},H_{\lambda}\}  =0 & \{P_{\xi},H_{\lambda}\}  =  P_{Y^{(a)}} -L_{ Y^{(a)} (\omega - \omega_0)_a} +H_{ Y^{(n)}} 
\end{eqnarray}
\end{subequations}
where $X= [c, \lambda e_n ]$, $Y = \mathrm{L}_{\xi}^{\omega_0} (\lambda e_n)$ and $Z^{(a)}$, $Z^{(n)}$ are the components of $Z\in\{X,Y\}$ with respect to the frame $(e_a, e_n)$.
\end{theorem}

\begin{remark}
This result improves on the results of \cite[Theorem 4.22]{CS2019}, since the constraints are manifestly independent of representatives of an equivalence class $[\omega]$, and because it allows us to present more explicit expressions for the Poisson brackets of constraints. Theorem \ref{thm:first-class-constraints} holds verbatim for higher dimensional generalisations of the theory as well (see Section \ref{s:generalisation}).
\end{remark}

\begin{proof}
The constraints are well defined on ${F}^{\partial}_{PC}$ because of the definition and properties of $\omega$ coming from Theorem \ref{thm:omegadecomposition}. 

In order to compute their Poisson brackets, we should first find their Hamiltonian vector fields.
We begin by varying the constraints. The variation of $\omega$ is constrained by \eqref{omegareprfix}. However, since \eqref{omegareprfix} imposes a constraint only on the part of $\omega$ in the kernel of $W_1^{\partial,(1,2)}$, it does not impose any condition on $e \delta \omega$. In the following computation we can always express the variation of the constraints in terms of $e\delta \omega$; hence the Hamiltonian vector fields are well defined and no other restriction has to be taken into account. The variations of $L_c$, $P_{\xi}$, $H_{\lambda}$ are respectively:
\begin{align*}
 \delta L_c= \int_{\Sigma}  -\frac{1}{2} c [\delta \omega, ee]  + \frac{1}{2} c d_{\omega}\delta (ee) = \int_{\Sigma} [c, e]  e \delta \omega  +  d_{\omega} c  e\delta e;
\end{align*}
\begin{align*}
 \delta P_{\xi} &= \int_{\Sigma}  \iota_{\xi} (e \delta e) F_{\omega}  -\frac{1}{2}\iota_{\xi} (e  e) d_{\omega}\delta \omega + \iota_{\xi} \delta \omega e d_{\omega} e  -\frac{1}{2} \iota_{\xi} (\omega-\omega_0) [\delta \omega, ee]\\
 & \qquad  + \frac{1}{2} \iota_{\xi} (\omega-\omega_0) d_{\omega}\delta (ee) \\
 & = \int_{\Sigma} - e \delta e \iota_{\xi}F_{\omega}  + \frac{1}{2}d_{\omega} \iota_{\xi} (e  e) \delta \omega - \frac{1}{2} \delta \omega \iota_{\xi} d_{\omega}(e e)  + \frac{1}{2}\delta \omega [ \iota_{\xi} (\omega-\omega_0), ee] \\
 & \qquad  + \frac{1}{2} d_{\omega} \iota_{\xi} (\omega-\omega_0) \delta (ee)\\
  & = \int_{\Sigma}  - e \delta e \iota_{\xi}F_{\omega}   -  (\mathrm{L}_{\xi}^{\omega} e ) e \delta \omega + e \delta \omega [ \iota_{\xi} (\omega-\omega_0), e]  +  d_{\omega} \iota_{\xi} (\omega-\omega_0) e \delta e\\ 
 & = \int_{\Sigma}  - e \delta e (\mathrm{L}_{\xi}^{\omega_0} (\omega-\omega_0) + \iota_ {\xi}F_{\omega_0})  -  (\mathrm{L}_{\xi}^{\omega_0} e ) e \delta \omega ;
\end{align*}
\begin{align*}
\delta H_{\lambda} &= \int_{\Sigma} \lambda e_n  \delta e  F_{\omega}+\frac{1}{2}\Lambda  \lambda e_n e^2 \delta e -\lambda e_n e  d_{\omega} \delta \omega\\ 
&= \int_{\Sigma} \lambda e_n  \delta e  F_{\omega} +\frac{1}{2}\Lambda  \lambda e_n e^2 \delta e+  d_{\omega}(\lambda e_n) e  \delta \omega + \lambda e_n d_{\omega} e  \delta \omega\\
&= \int_{\Sigma} \lambda e_n  \delta e  F_{\omega} +\frac{1}{2}\Lambda  \lambda e_n e^2 \delta e+  d_{\omega}(\lambda e_n) e  \delta \omega + \lambda \sigma  e \delta \omega.
\end{align*}
In the last computation we used \eqref{omegareprfix} with $\sigma := {W_1^{\partial,(1,1)}}^{-1} (e_n d_{\omega} e )$. 
Hence, the components of the Hamiltonian vector fields of $L_c$ and $P_\xi$ are
\begin{eqnarray}
&\mathbb{L}_e = [c,e] &  \mathbb{L}_\omega = d_{\omega} c \label{e:ham_vf_J}\\
&\mathbb{P}_e = - \mathrm{L}_{\xi}^{\omega_0} e & \mathbb{P}_\omega = - \mathrm{L}_{\xi}^{\omega_0} (\omega-\omega_0) - \iota_ {\xi}F_{\omega_0}
\label{e:ham_vf_K}
\end{eqnarray}
where, e.g., $\mathbb{L}_e \equiv \mathbb{L}(e)$, with $\iota_{\mathbb{L}}\varpi_{PC}= \delta L_c$. The components of the Hamiltonian vector field of $H_\lambda$ are described by
\begin{equation}\label{e:ham_vf_L}
\mathbb{H}_e= d_{\omega}(\lambda e_n) + \lambda \sigma  \qquad e \mathbb{H}_\omega =  \lambda e_n F_{\omega}+\frac{1}{2}\Lambda  \lambda e_n e^2.
\end{equation}
The second equation, together with the requirement that $\mathbb{H}$ preserves the structural constraint \eqref{omegareprfix}, uniquely determines $\mathbb{H}_\omega$. However, we do not need an explicit expression for it, since in the computations we will only need $e\mathbb{H}_\omega$.
We can now compute the brackets between the constraints:
\begin{align*}
\{L_c, L_c\} & = \int_{\Sigma}  [c,e] e  d_{\omega} c = \int_{\Sigma}  \frac{1}{2} [c,ee]  d_{\omega} c \\
 &= \int_{\Sigma} \frac{1}{4} d_{\omega}[c , c] ee = \int_{\Sigma} -\frac{1}{2} [c , c]  e d_{\omega}e = - \frac{1}{2} L_{[c,c]};
\end{align*}

\begin{align*}
\{L_c, P_{\xi}&\} = \int_{\Sigma} - [c,e] e(\mathrm{L}_{\xi}^{\omega_0} (\omega-\omega_0) + \iota_ {\xi}F_{\omega_0}) - d_{\omega} c e  \mathrm{L}_{\xi}^{\omega_0} e \\
& = \int_{\Sigma} \frac{1}{2} \left(\mathrm{L}_{\xi}^{\omega_0}c [\omega- \omega_0, ee]+ c [\omega- \omega_0,\mathrm{L}_{\xi}^{\omega_0}( ee)]- c [ee, \iota_ {\xi}F_{\omega_0}]-  d_{\omega} \mathrm{L}_{\xi}^{\omega_0} (e e) c\right)\\
& = \int_{\Sigma} \frac{1}{2} \mathrm{L}_{\xi}^{\omega_0}c [\omega, ee]- \frac{1}{2} d c \iota_{\xi} d(ee) + \frac{1}{2} [ \iota_{\xi}\omega_0, d(ee)] c\\
& =  \int_{\Sigma} \frac{1}{2}  \mathrm{L}_{\xi}^{\omega_0}c d_\omega(ee) = \int_{\Sigma} \mathrm{L}_{\xi}^{\omega_0}c e d_\omega e = L_{\mathrm{L}_{\xi}^{\omega_0}c};
\end{align*}

\begin{align*}
\{P_{\xi}, P_{\xi}&\}  = \int_{\Sigma} \frac{1}{2}  \mathrm{L}_{\xi}^{\omega_0} (e e )\mathrm{L}_{\xi}^{\omega_0} (\omega - \omega_0) +  \frac{1}{2}  \mathrm{L}_{\xi}^{\omega_0} (e e ) \iota_{\xi}F_{\omega_0}\\
& \overset{\diamondsuit\clubsuit}{=} \int_{\Sigma} \frac{1}{4} \mathrm{L}_{[\xi,\xi]}^{\omega_0}(e e ) (\omega - \omega_0) + \frac{1}{4}[\iota_{\xi}\iota_{\xi}F_{\omega_0},e e ] (\omega - \omega_0) +  \frac{1}{2}  \mathrm{L}_{\xi}^{\omega_0} (e e ) \iota_{\xi}F_{\omega_0}\\
& = \int_{\Sigma} \frac{1}{4} \iota_{[\xi,\xi]}d_{\omega_0}(e e ) (\omega - \omega_0)+ \frac{1}{4}d_{\omega_0}\iota_{[\xi,\xi]}(e e ) (\omega - \omega_0) \\
& \qquad + \frac{1}{4}[\iota_{\xi}\iota_{\xi}F_{\omega_0},e e ] (\omega - \omega_0) +  \frac{1}{2}  \mathrm{L}_{\xi}^{\omega_0} (e e ) \iota_{\xi}F_{\omega_0}\\
& \overset{\diamondsuit}{=} \int_{\Sigma} \frac{1}{4} \iota_{[\xi,\xi]}d_{\omega}(e e ) (\omega - \omega_0)-\frac{1}{4} \iota_{[\xi,\xi]}[\omega - \omega_0,e e ] (\omega - \omega_0) \\
& \qquad+ \frac{1}{4}\iota_{[\xi,\xi]}(e e ) d_{\omega_0}(\omega - \omega_0) + \frac{1}{4}[\iota_{\xi}\iota_{\xi}F_{\omega_0},e e ] (\omega - \omega_0) +  \frac{1}{2}  \mathrm{L}_{\xi}^{\omega_0} (e e ) \iota_{\xi}F_{\omega_0}\\
\intertext{}
& \overset{\heartsuit}{=} \int_{\Sigma} \frac{1}{4} d_{\omega}(e e ) \iota_{[\xi,\xi]}(\omega - \omega_0)-\frac{1}{4} [\omega - \omega_0,e e ] \iota_{[\xi,\xi]}(\omega - \omega_0)- \frac{1}{4}\iota_{[\xi,\xi]}(e e ) F_{\omega_0} \\
& \qquad+\frac{1}{4}\iota_{[\xi,\xi]}(e e )F_{\omega} -\frac{1}{8}\iota_{[\xi,\xi]}(e e )[\omega_0-\omega,\omega_0-\omega] \\
& \qquad + \frac{1}{4}[\iota_{\xi}\iota_{\xi}F_{\omega_0},e e ] (\omega - \omega_0) +  \frac{1}{2}  \mathrm{L}_{\xi}^{\omega_0} (e e ) \iota_{\xi}F_{\omega_0}\\
& \overset{\spadesuit}{=} \int_{\Sigma} \frac{1}{4} d_{\omega}(e e ) \iota_{[\xi,\xi]}(\omega - \omega_0)+\frac{1}{4}\iota_{[\xi,\xi]}(e e )F_{\omega}+ \frac{1}{4}d_{\omega_0}(e e )\iota_{\xi}\iota_{\xi} F_{\omega_0} \\
& \qquad + \frac{1}{2}d_{\omega_0}\iota_{\xi}(e e ) \iota_{\xi} F_{\omega_0}- \frac{1}{4}\iota_{\xi}\iota_{\xi}F_{\omega_0} [\omega - \omega_0,e e]  \\
& \qquad +  \frac{1}{2} \left( \iota_{\xi}d_{\omega_0} (e e )-  d_{\omega_0} \iota_{\xi} (e e ) \right) \iota_{\xi}F_{\omega_0}\\
& = \int_{\Sigma} \frac{1}{4} d_{\omega}(e e ) \iota_{[\xi,\xi]}(\omega - \omega_0)+\frac{1}{4}\iota_{[\xi,\xi]}(e e )F_{\omega}-\frac{1}{4}d_{\omega}(e e )\iota_{\xi}\iota_{\xi} F_{\omega_0} \\
& = \frac{1}{2}P_{[\xi, \xi]} - \frac{1}{2}L_{\iota_{\xi}\iota_{\xi}F_{\omega_0}}.
\end{align*}
In these computations we used integration by parts ($\diamondsuit$) and the following identities (for a proof of the second see \cite[Lemma 18]{CS2017}):
\begin{align*}
(\spadesuit) \qquad \frac{1}{2}\iota_{[\xi,\xi]}A &= - \frac{1}{2} \iota_{\xi}\iota_{\xi} d_{\omega_0}A + \iota_{\xi}d_{\omega_0}\iota_{\xi} A- \frac{1}{2} d_{\omega_0} \iota_{\xi}\iota_{\xi} A \qquad \forall A \in \Omega^{i,j}_{\partial} \\
(\clubsuit) \qquad \mathrm{L}_{\xi}^{\omega_0}\mathrm{L}_{\xi}^{\omega_0}B &= \frac{1}{2}\mathrm{L}_{[\xi,\xi]}^{\omega_0}B + \frac{1}{2}[\iota_{\xi}\iota_{\xi}F_{\omega_0},B] \qquad \forall B \in \Omega^{i,j}_{\partial}\\
(\heartsuit) \qquad d_{\omega_0}(\omega_0-\omega)&= F_{\omega_0} -F_{\omega} +\frac{1}{2}[\omega_0-\omega,\omega_0-\omega].
\end{align*}
\begin{align*}
\{L_c,  H_{\lambda}\} & = \int_{\Sigma} [c,e] \lambda e_n F_{\omega} +\frac{1}{2}[c,e]\Lambda  \lambda e_n e^2 + d_{\omega} c e (d_{\omega}(\lambda e_n) + \lambda \sigma) \\ 
& = \int_{\Sigma} [c,e] \lambda e_n F_{\omega} +\frac{1}{3!}[c,e^3]\Lambda  \lambda e_n + d_{\omega} c  d_{\omega}(\lambda e_n e) \\
&=  \int_{\Sigma} - [c, \lambda e_n ] eF_{\omega}-\frac{1}{3!}\Lambda[c, \lambda e_n ] e^3 \\
& = \int_{\Sigma} -[c, \lambda e_n ]^{(a)}e_a e F_{\omega} -[c, \lambda e_n ]^{(n)}e_n e F_{\omega}-\frac{1}{3!}\Lambda[c, \lambda e_n ]^{(n)}e_n e^3 \\
& = - P_{[c, \lambda e_n ]^{(a)}} + L_{[c, \lambda e_n ]^{(a)}(\omega - \omega_0)_a} - H_{[c, \lambda e_n ]^{(n)}};
\end{align*}
Finally we have
\begin{align*}
\{H_{\lambda},H_{\lambda}\} & = \int_{\Sigma} (d_{\omega}(\lambda e_n) + \lambda \sigma) \left(\lambda e_n F_{\omega}+\frac{1}{2}\Lambda  \lambda e_n e^2\right)=\\
& = \int_{\Sigma}  d_{\omega}\lambda e_n \left(\lambda e_n F_{\omega}+\frac{1}{2}\Lambda  \lambda e_n e^2\right)- \lambda d_{\omega}e_n\left(\lambda e_n F_{\omega}+\frac{1}{2}\Lambda  \lambda e_n e^2\right) =0,
\end{align*}
since $\lambda \lambda =0$ and $e_n e_n=0$, and
\begin{align*}
\{P_{\xi},H_{\lambda}\} & = \int_{\Sigma} - \mathrm{L}_{\xi}^{\omega_0} e \lambda e_n F_{\omega} -\frac{1}{2}\Lambda \mathrm{L}_{\xi}^{\omega_0} e \lambda e_n e^2\\
& \qquad- \left( \mathrm{L}_{\xi}^{\omega_0} (\omega-\omega_0)+ \iota_ {\xi}F_{\omega_0}\right)e(d_{\omega}(\lambda e_n) + \lambda \sigma)   \\
& = \int_{\Sigma} - \mathrm{L}_{\xi}^{\omega_0} e \lambda e_n F_{\omega}-\frac{1}{3!}\Lambda \mathrm{L}_{\xi}^{\omega_0} e^3 \lambda e_n\\
& \qquad  -\left( \mathrm{L}_{\xi}^{\omega_0} (\omega-\omega_0)+ \iota_ {\xi}F_{\omega_0}\right) d_{\omega}(e \lambda e_n) \\
\intertext{}
& = \int_{\Sigma}  e  \mathrm{L}_{\xi}^{\omega_0} (\lambda e_n)  F_{\omega}+\frac{1}{3!}\Lambda  e^3 \mathrm{L}_{\xi}^{\omega_0}(\lambda e_n)  + e   \lambda e_n \mathrm{L}_{\xi}^{\omega_0} F_{\omega}\\
& \qquad + \left( d_{\omega} \iota_{\xi} (\omega-\omega_0)- \iota_ {\xi}F_{\omega}\right) d_{\omega}(e \lambda e_n)  \\
& = \int_{\Sigma}  e  \mathrm{L}_{\xi}^{\omega_0} (\lambda e_n)  F_{\omega} +\frac{1}{3!}\Lambda  e^3 \mathrm{L}_{\xi}^{\omega_0}(\lambda e_n)+ e   \lambda e_n \mathrm{L}_{\xi}^{\omega_0} F_{\omega}\\
& \qquad - [F_{\omega}, \iota_{\xi} (\omega-\omega_0)]e \lambda e_n- \mathrm{L}_ {\xi}^{\omega} F_{\omega} e \lambda e_n \\
& = \int_{\Sigma}    \mathrm{L}_{\xi}^{\omega_0} (\lambda e_n) e F_{\omega} +\frac{1}{3!}\Lambda  e^3 \mathrm{L}_{\xi}^{\omega_0}(\lambda e_n) \\
& = P_{ \mathrm{L}_{\xi}^{\omega_0} (\lambda e_n)^{(a)}} +H_{ \mathrm{L}_{\xi}^{\omega_0} (\lambda e_n)^{(n)}}-L_{ \mathrm{L}_{\xi}^{\omega_0} (\lambda e_n)^{(a)} (\omega - \omega_0)_a},
\end{align*}
where we used that $  \mathrm{L}_{\xi}^{\omega_0} F_{\omega} - \mathrm{L}_{\xi}^{\omega}F_{\omega} = [\iota_{\xi}(\omega_0-\omega),  F_{\omega}]$.
This shows that the relations \eqref{brackets-of-constraints} hold and, therefore, that the constraints are first class.

\end{proof}
\begin{remark}
Theorem \ref{thm:first-class-constraints}, in particular, shows that on time-like or space-like boundaries the constraints \eqref{constraints} are first class. Counting the number of components of the Lagrange multipliers $c$, $\xi$ and $\lambda$ we deduce that there are 10 local constraints, while the number of independent components of the conjugate fields $e$ and $\omega$ is 12. Hence we recover the classical result of having 2 local physical degrees of freedom.
\end{remark}

\begin{remark}
From the expressions of the Hamiltonian vector fields of the constraints \eqref{e:ham_vf_J}, \eqref{e:ham_vf_K} and \eqref{e:ham_vf_L} we deduce that the constraint $L_c$ describes the action of the gauge transformations of the theory, while the constraints $P_\xi$ and $H_\lambda$ describe the action of the diffeomorphisms, respectively tangent and transversal to the boundary.
\end{remark}

\section{Palatini--Cartan theory and its BFV data}\label{s:BFV}
This section is not required by Section \ref{s:RPSext} and can therefore be skipped by readers that are not interested in the BFV formalism but only wish to see the higher dimensional generalisation of the construction of the reduced phase space. It will be however required  by Section \ref{s:BFVext}, where we will discuss the higher dimensional version of the BFV formalism.

\subsection{From the reduced phase space to BFV}\label{sec:from_KT_to_BFV}
\newcommand{\uC}{\underline{C}}
\newcommand{\uomega}{\underline{\varpi}}
We start with a short overview of the BFV formalism. The problem we wish to address is the symplectic reduction of a coisotropic submanifold. For simplicity, as this is also the case at hand in this paper, we will consider only the situation where the submanifold is defined in terms of global constraints\footnote{This in general requires an appropriate version of the implicit function theorem, but we will effectively work within an algebraic setting.}. More precisely, 
the starting point are a symplectic manifold $(M,\varpi)$ and a collection $\{\phi_i\}$ of independent, differentially independent constraints; their common zero locus $C=\{x\in M:\phi_i(x)=0\ \forall i\}$ is then a 
submanifold.\footnote{For notational simplicity, we assume here a discrete family of constraints, even though in the case of field theory we will need a continuos family. In that case the sums will be replaced by integrals.} In addition, the constraints are assumed to be of first class: i.e., their Poisson brackets vanish on $C$, or, equivalently, they satisfy
\begin{equation}\label{e:fijk}
\{\phi_i,\phi_j\}=f_{ij}^k \phi_k,
\end{equation}
where the $\{f_{ij}^k\}$s are functions on $M$ (we assume a sum over repeated indices). The restriction of $\varpi$ to $C$ becomes degenerate, but one can easily show that its kernel, called the characteristic distribution, is spanned by the Hamiltonian vector fields $X_i$ of the constraints $\phi_i$. As a consequence of \eqref{e:fijk}, the characteristic distribution is involutive. 
The symplectic reduction $\uC$ of $C$ is the quotient, which we temporarily assume to be smooth, of $C$ by its characteristic distribution, endowed with the unique symplectic form $\uomega$ whose pullback to $C$ is the restriction of $\varpi$. 

Since $\uC$ is very often not smooth in applications, it is better to resort to a different, more flexible description. Algebraically, we have that $C^\infty(\uC)=C^\infty(C)^\text{inv}$, where inv means invariant under the Hamiltonian vector fields $X_i$. In turn, $C^\infty(C)=C^\infty(M)/I$, where 
$I=\text{span}_{C^\infty(M)}\{\phi_i\}$ is the vanishing ideal of $C$. Therefore, we have 
$C^\infty(\uC)=(C^\infty(M)/I)^\text{inv}$. One can show that this algebra inherits a Poisson bracket which, in the smooth case, is also the one induced by $\uomega$. The Poisson algebra $(C^\infty(M)/I)^\text{inv}$ is defined also if $\uC$ is not smooth and it may be tempting to take it as a good replacement for $\uC$. The problem is that often this algebra is very poor (for example, if $\uC$ is not Hausdorff, this algebra is just $\mathbb{R}$).

A better way to proceed is to look for a cohomological description of the symplectic quotient. This is what is achieved by the BFV formalism. Namely, one first adds new odd variables $c^i$ of degree (ghost number) $+1$, called the ghosts, one for each constraint, and their momenta $c^\dagger_i$ (a.k.a.\ the antighosts),which are also odd and have degree $-1$. One extends the original symplectic manifold $(M,\varpi)$ to a graded symplectic manifold
$M\times T^*W$, where $W$ is the odd vector space whose coordinates are the $c^i$s, with symplectic form
\[
\varpi \to \varpi +\int_\Sigma \delta c^\dagger_i\, \delta c^i.
\]
Next one introduces the BFV action, an odd function of degree $1$,
\[
S = \int_\Sigma c^i\phi_i + \frac12 f_{ij}^k c^\dagger_k c^ic^j + R,
\]
where $R$ is a function of higher degree in the ghost momenta $c^\dag_i$ such that $\{S,S\}=0$ (the BFV master equation). It has been proved \cite{BV1981,Batalin:1983,Stasheff1997} that one can always find such a correction $R$. The Hamiltonian vector field $Q$ of $S$
is odd, of degree $1$, and satisfies $[Q,Q]=0$ (such a vector field is called cohomological because it acts as a differential on the algebra of functions). We have
\begin{align*}
Qc^\dagger_i &= \phi_i+\cdots,\\
Qf &= c^iX_i(f) + \cdots,
\end{align*}
where $f$ is a function on $M$ and $\cdots$ denotes terms depending on the ghost momenta. {}From this we see that, up to these higher terms, the image of $Q$ contains the vanishing ideal $I$ and the kernel of $Q$ selects the invariant functions.
One can actually show \cite{BV1981,Batalin:1983,Stasheff1997}  that in degree zero there is not more than this: The cohomology of $Q$ in degree zero is isomorphic to $(C^\infty(M)/I)^\text{inv}$ as a Poisson algebra. The idea of the BFV formalism is then to replace the original, possibly singular symplectic reduction with the ``BFV manifold''
\[
(M\times T^*W, \varpi + \int_{\Sigma} \delta c^\dagger_i\, \delta c^i, S).
\]
The complex $(C^\infty(M\times T^*W),Q)$ is the sought for cohomological resolution of the symplectic reduction of $C$.\footnote{Indeed, $Q$ is a deformation of a combination of the Koszul--Tate complex, which gives a cohomological resolution of $C$ as a submanifold of $M$, and of the Chevalley--Eilenberg complex of the Lie algebroid naturally associated to the conormal bundle of $C$. This deformation is compatible with the symplectic structure.}

Note that there is some freedom in the construction of the BFV data, but one can show \cite{Stasheff1997} that the solution is unique up to symplectomorphisms compatible with the BFV actions. A particularly good solution is when the correction term $R$ vanishes. This is not always possible, but it is so in some cases. The most important one is when one can choose the constraints in such a way that the $\{f_{ij}^k\}$s are constant (this means that the constraints are assembled into an equivariant momentum map and that the reduction is actually an example of Marsden--Weinstein reduction). In this case the BFV construction goes often under the name of BRS \cite{KS1987}. %

It may however happen that the correction $R$ vanishes beyond the BRS case. We will see that this is actually 
what occurs in the PC case at hand. A similar phenomenon was observed in the BFV treatment of Einstein--Hilbert gravity \cite{CS2016b}. 

\begin{remark}
The BFV formalism was introduced by  Batalin and Vilkovisky in \cite{BV1981} and by Batalin and Fradkin \cite{Batalin:1983}. Stasheff \cite{Stasheff1997} gave a mathematical treatment with formal proofs of existence and uniqueness, based on homological perturbation theory, and treated a more general case based on Lie--Rinehart algebras. Sch\"atz \cite{Schaetz:2008} extended the result to general coisotropic submanifolds, not necessarily given in terms of constraints. In \cite{CMR2012} 
the relation between the BV formalism in the bulk of a field theory with its BFV formalism on the boundary was clarified; in \cite{CMR2012b} a procedure to recover the BFV boundary data from the BV bulk data was given; several examples, including Yang--Mills, Chern--Simons and $BF$ theory were treated. The case of Einstein--Hilbert gravity was successfully treated in \cite{CS2016b}. However, in \cite{CS2017} it was shown that the natural implementation of the BV bulk formalism for four-dimensional Palatini--Cartan theory does lead to singular BFV boundary data.
\end{remark}

\begin{remark}
The BFV formalism is not only introduced to provide a cohomological resolution of possibly singular symplectic reductions, but also as a way to perform quantisation. The idea is to quantise the extended graded symplectic space $M\times T^*W$ to some graded Hilbert space (which may be reasonably easy, since often $M$ is itself a cotangent bundle) and to find an operator $\Hat S$ that quantises the BFV action $S$ and that satisfies $[\Hat S,\Hat S]=0$. The master equation $\{S,S\}=0$ ensures that this is possible at the lowest order in $\hbar$. If one can achieve this condition at all orders, then one can define the Hilbert space that quantises the symplectic reduction as the cohomology in degree zero of $\Hat S$. There may be obstructions (anomalies) to achieve this program. In \cite{CMR2} a procedure was introduced that, when successful, allows constructing the operator $\Hat S$ from the perturbative quantisation of the bulk BV data and, at the same time,
a state for the bulk theory in the cohomology of such operator.
\end{remark}

\subsection{BFV Structure of Palatini--Cartan theory} \label{sec:BFV-action}
From the constraints and their brackets it is possible extend the space of fields to a graded symplectic manifold by promoting the Lagrange multipliers to ghosts and adding ghost momenta. 
The following Theorem \ref{thm:BFVaction} shows that the na\"ive guess for a BFV action, containing only the constraints (constant term in the ghost momenta) and the information on their Poisson brackets (linear term in the ghost momenta) already satisfies the BFV master equation. 
\begin{remark}
At a physical level, the Lagrange multipliers assume the meaning of symmetry generators of the system. In particular the field $c \in \Omega^{0,2}_\partial$ represents the internal gauge symmetry (recall that we are using the identification $\mathfrak{so}(3,1) \cong \wedge^2 \mathcal{V}$); the vector field $\xi \in \mathfrak{X}(\Sigma)$ represents the vector fields parametrising local diffeomorphisms tangent to the boundary; the scalar field $\lambda \in C^{\infty}[1](\Sigma)$ might in turn be thought of as the parameter representing the local diffeomorphisms in the transversal direction. This becomes evident when considering the classical part of the cohomological vector field $Q$ (see Equation \eqref{Q_boundary_part}, below).
\end{remark}

\begin{theorem}\label{thm:BFVaction}
Let $g^\partial$ be nondegenerate on $\Sigma$. Let $\mathcal{F}$ be the bundle

\begin{equation}
\mathcal{F} \longrightarrow \Omega_{nd}^1(\Sigma, \mathcal{V}),
\end{equation}
with local trivialisation on an open $\mathcal{U}_{\Sigma} \subset \Omega_{nd}^1(\Sigma, \mathcal{V})$
\begin{equation}\label{LoctrivF1}
\mathcal{F}\simeq \mathcal{U}_{\Sigma} \times \mathcal{A}^{red}(\Sigma) \oplus T^* \left(\Omega_{\partial}^{0,2}[1]\oplus \mathfrak{X}[1](\Sigma) \oplus C^\infty[1](\Sigma)\right),
\end{equation}
and fields denoted by $e \in \mathcal{U}_{\Sigma}$ and $\omega \in \mathcal{A}^{red}(\Sigma)$ in degree zero, $c \in\Omega_{\partial}^{0,2}[1]$, $\xi \in\mathfrak{X}[1](\Sigma)$ and $\lambda\in \Omega^{0,0}[1]$ in degree one, $c^\dag\in\Omega_{\partial}^{3,2}[-1]$, $\lambda^\dag\in\Omega_{\partial}^{3,4}[-1]$ and $\xi^\dag\in\Omega_\partial^{1,0}[-1]\otimes\Omega_{\partial}^{3,4}$ in degree minus one, together with a fixed  $e_n \in \Gamma(\mathcal{V})$, completing the image of elements $e \in\mathcal{U}_{\Sigma}$ to a basis of  $\mathcal{V}$;
define a symplectic form and an action functional on $\mathcal{F}$ respectively by
\begin{align}
\varpi = \int_{\Sigma} &e \delta e \delta \omega + \delta c \delta c^\dag + \delta \lambda \delta \lambda^\dag + \iota_{\delta \xi} \delta \xi^\dag,\label{symplectic_form_NC1} \\
S= \int_{\Sigma} & c e d_{\omega} e + \iota_{\xi} e e F_{\omega} + \iota_{\xi} (\omega-\omega_0) e d_{\omega} e+ \lambda e_n \left(eF_\omega + \frac{1}{3!}\Lambda e^3\right) +\frac{1}{2} [c,c] c^{\dag}\nonumber\\ 
& - \mathrm{L}_{\xi}^{\omega_0} c c^{\dag} +\frac{1}{2}\iota_{\xi}\iota_{\xi}F_{\omega_0}c^{\dag}  + [c, \lambda e_n ]^{(a)}(\xi_a^{\dag}- (\omega - \omega_0)_a c^\dag) + [c, \lambda e_n ]^{(n)}\lambda^\dag \nonumber\\
 &- \mathrm{L}_{\xi}^{\omega_0} (\lambda e_n)^{(a)}(\xi_a^{\dag}- (\omega - \omega_0)_a c^\dag) - \mathrm{L}_{\xi}^{\omega_0} (\lambda e_n)^{(n)}\lambda^\dag - \frac{1}{2}\iota_{[\xi,\xi]}\xi^{\dag} \label{action_NC1}
\end{align}
where $e$ and $\omega$ satisfy the additional requirement $e_n d_{\omega} e \in \Ima W_1^{\partial,(1,1)}$.
Then the triple $(\mathcal{F}, \varpi, S)$ defines a BFV structure on $\Sigma$.\end{theorem}

\begin{proof} We have to prove that the action $S$ satisfies the classical master equation.
 By definition we have
\begin{align*}
\{S,S\}= \iota_Q \iota_Q \varpi.
\end{align*}
where $Q$ is the Hamiltonian vector field of $S$, defined by $\iota_Q \varpi = \delta S$.

In order to simplify the computation we can divide the action in two parts: $$S= S_0 +S_1$$ where $S_0$ is independent of the ghost momenta and $S_1$ is linear in them. In particular $S_0$ is the sum of the constraints and $S_1$ is everything else.
We divide the symplectic form too: $$\varpi= \varpi_f + \varpi_g$$ where $\varpi_f= \int_{\Sigma} e \delta e \delta \omega$ is the classical part and $\omega_g=  \int_{\Sigma}  \delta c \delta c^\dag + \delta \lambda \delta \lambda^\dag + \iota_{\delta \xi} \delta \xi^\dag$ is the ghost part. Finally, we define $Q_0$ to be the part of $Q$ satisfying $\iota_{Q_0} \varpi = \delta S_0$ and $Q_1$ to be the one satisfying $\iota_{Q_1} \varpi = \delta S_1$.

We can divide the master equation into the corresponding parts:
\begin{align*}
\{S,S\}=\{S_0,S_0\}_f+2\{S_0,S_1\}_f+2\{S_0,S_1\}_g+\{S_1,S_1\}_f+\{S_1,S_1\}_g
\end{align*}
where 
\begin{subequations}\label{CME-parts}
\begin{eqnarray}
\{S_0,S_0\}_f= \iota_{Q_0}\iota_{Q_0}\varpi_f & \{S_0,S_1\}_f= \iota_{Q_0}\iota_{Q_1}\varpi_f \\
\{S_0,S_0\}_g= \iota_{Q_0}\iota_{Q_0}\varpi_g & \{S_0,S_1\}_g= \iota_{Q_0}\iota_{Q_1}\varpi_g \\
\{S_1,S_1\}_f= \iota_{Q_1}\iota_{Q_1}\varpi_f & \{S_1,S_1\}_g= \iota_{Q_1}\iota_{Q_1}\varpi_g. 
\end{eqnarray}
\end{subequations}
This subdivision is particularly convenient, since we can exploit some properties of the action and prove the master equation piecewise.
We first note that $\{S_0,S_0\}_g=0$ since $S_0$ has no antighost part. Furthermore, by Theorem \ref{thm:BFVaction} we have that 
 \begin{align*}
 \{S_0,S_0\}_f+2\{S_0,S_1\}_g=0.
 \end{align*}
 The terms $\{S_0,S_1\}_f$ and $\{S_1,S_1\}_g$ are linear in the antighost while $\{S_1,S_1\}_f$ is quadratic in the antighost. Hence we should prove separately that $2\{S_0,S_1\}_f+\{S_1,S_1\}_g=0$ and $\{S_1,S_1\}_f=0$. For these last two terms we have to do the computation explicitly. We start by computing $\delta S$ in order to get $Q$ from the equation $\iota_Q \varpi = \delta S$.
 
 Note that for $X$ odd, since $\delta e_n =0$, we have
 \begin{align*}
\delta X = \delta (X^{(\mu)} e_{\mu})  &=  \delta (X^{(\mu)}) e_{\mu} + X^{(a)} \delta ( e_{a})\\
\delta (X^{(\mu)}) &= (\delta X)^{(\mu)}-  X^{(a)} \delta ( e_{a})^{(\mu)}.
 \end{align*}
The variation of the action is 
\begin{align*}
\delta S= \int_{\Sigma} & \delta c e d_{\omega} e - \frac{1}{2}c [\delta \omega , ee] +d_{\omega}c e \delta e   +\frac{1}{2} \iota_{\delta \xi} (e e) F_{\omega} + \iota_{\delta \xi} (\omega-\omega_0) e d_{\omega} e \nonumber\\ 
&  - e \delta e( \mathrm{L}_{\xi}^{\omega_0} (\omega-\omega_0) + \iota_ {\xi}F_{\omega_0}) -  (\mathrm{L}_{\xi}^{\omega_0} e ) e \delta \omega + \delta \lambda e_n e F_{\omega}+\lambda e_n  \delta e  F_{\omega}+ \nonumber\\ 
&  \frac{1}{3!}\Lambda \delta \lambda e_n   e^3+  \frac{1}{2}\Lambda \lambda e_n   e^2 \delta e +  d_{\omega}(\lambda e_n) e  \delta \omega + \lambda \sigma  e \delta \omega + [\delta c,c] c^{\dag}+ \frac{1}{2}[c,c] \delta c^{\dag} \nonumber\\ 
\intertext{}
& - \iota_{\delta \xi} d_{\omega_0} c c^{\dag} +  \delta c  d_{\omega_0} \iota_{\xi} c^{\dag} - \mathrm{L}_{\xi}^{\omega_0} c \delta c^{\dag} +\iota_{\delta \xi}\iota_{\xi}F_{\omega_0}c^{\dag} +\frac{1}{2}\iota_{\xi}\iota_{\xi}F_{\omega_0}\delta c^{\dag}\nonumber\\ 
 & + \left([\delta c, \lambda e_n ]^{(a)}- [c, \delta \lambda e_n ]^{(a)}  - [c, \lambda e_n ]^{(b)}\delta e_b^{(a)}\right)(\xi_a^{\dag}- (\omega - \omega_0)_a c^\dag) \nonumber\\ 
 & + [c, \lambda e_n ]^{(a)}(\delta \xi_a^{\dag}- \delta (\omega - \omega_0)_a c^\dag- (\omega - \omega_0)_a \delta c^\dag)  + [\delta c, \lambda e_n ]^{(n)}\lambda^\dag \nonumber\\ 
 & - [c, \delta \lambda e_n ]^{(n)}\lambda^\dag  - [c, \lambda e_n ]^{(b)}\delta e_b^{(n)}\lambda^\dag + [c, \lambda e_n ]^{(n)}\delta \lambda^\dag \nonumber\\
 &\left(- (\iota_{\delta \xi} d_{\omega_0}(\lambda e_n))^{(a)} + \mathrm{L}_{\xi}^{\omega_0} (\delta \lambda e_n)^{(a)}   +\mathrm{L}_{\xi}^{\omega_0} (\lambda e_n)^{(b)}\delta e_b^{(a)}\right)(\xi_a^{\dag}- (\omega - \omega_0)_a c^\dag)\nonumber\\
 &  - \mathrm{L}_{\xi}^{\omega_0} (\lambda e_n)^{(a)}(\delta \xi_a^{\dag}- \delta (\omega - \omega_0)_a c^\dag- (\omega - \omega_0)_a \delta c^\dag) \nonumber\\
 & - (\iota_{\delta \xi} d_{\omega_0}(\lambda e_n))^{(n)}\lambda^\dag  + \mathrm{L}_{\xi}^{\omega_0} (\delta \lambda e_n)^{(n)}\lambda^\dag  +\mathrm{L}_{\xi}^{\omega_0} (\lambda e_n)^{(b)}\delta e_b^{(n)}\lambda^\dag\nonumber\\
 &- \mathrm{L}_{\xi}^{\omega_0} (\lambda e_n)^{(n)}\delta \lambda^\dag  - \delta \xi^a (\partial_a \xi^b) \xi_b^{\dag}- \delta \xi^a \partial_b (\xi^b \xi_a^{\dag})- \xi^a (\partial_a \xi^b) \delta \xi_b^{\dag}.
\end{align*} 
This variation contains all the information necessary to construct the cohomological vector field $Q$. However $\delta S$ contains some variation of $\delta \omega$ that are constrained by \eqref{omegareprfix} and some other terms of difficult explicit inversion. For our purposes it is sufficient to have the explicit expressions of
$Q_{0e}, Q_{0 \omega}, Q_{c}, Q_{\lambda}, Q_{\xi}$ and some information about $Q_{1e}, Q_{1\omega}$ (recall that $Q_{0e}, Q_{0 \omega} $ are the part of $Q_{e}, Q_{\omega} $ not containing antighosts, while $Q_{1e}, Q_{1\omega}$ contain everything else.

Let us start from $Q_{1e}, Q_{1\omega}$. They are defined through the equation 
\begin{align*}
\iota_{Q_1}(e \delta e \delta \omega) = &  - [c, \lambda e_n ]^{(b)}\delta e_b^{(a)}(\xi_a^{\dag}- (\omega - \omega_0)_a c^\dag) - [c, \lambda e_n ]^{(a)}\delta (\omega - \omega_0)_a c^\dag \\
 & - [c, \lambda e_n ]^{(b)}\delta e_b^{(n)}\lambda^\dag +\mathrm{L}_{\xi}^{\omega_0} (\lambda e_n)^{(b)}\delta e_b^{(a)}(\xi_a^{\dag}- (\omega - \omega_0)_a c^\dag) \\
 &+ \mathrm{L}_{\xi}^{\omega_0} (\lambda e_n)^{(a)}\delta (\omega - \omega_0)_a c^\dag +\mathrm{L}_{\xi}^{\omega_0} (\lambda e_n)^{(b)}\delta e_b^{(n)}\lambda^\dag
\end{align*}
Since $\lambda$ is a scalar function we have that $[c, \lambda e_n ]^{(a)}=\lambda [c,  e_n ]^{(a)}$ and
$\mathrm{L}_{\xi}^{\omega_0} (\lambda e_n)^{(a)}= \mathrm{L}_{\xi}^{\omega_0} (\lambda) e_n^{(a)}-\lambda \mathrm{L}_{\xi}^{\omega_0} ( e_n)^{(a)}= -\lambda \mathrm{L}_{\xi}^{\omega_0} ( e_n)^{(a)}$, since $e_n^{(a)}=0$. We  then deduce that  every term in $Q_{1e}$ and $Q_{1\omega}$ must be linear in $\lambda$. From \eqref{CME-parts} we have
\begin{align*}
\{S_1,S_1\}_f= \iota_{Q_1}\iota_{Q_1}(e \delta e \delta \omega)= 2 e Q_{1e}Q_{1\omega}
\end{align*}
which contains only terms proportional to $\lambda^2=0$ since $\lambda$ is an odd scalar function. This proves 
$\{S_1,S_1\}_f=0$.

From the above variation of $S$ we can compute directly $Q_{0e}, Q_{0 \omega}, Q_{c}, Q_{\lambda}, Q_{\xi}$ :
\begin{equation}\label{Q_boundary_part}
\begin{split}
Q_0e&= [c,e] - \mathrm{L}_{\xi}^{\omega_0} e +  d_{\omega}(\lambda e_n) + \lambda \sigma \\
Q_0\omega &= d_{\omega}c - \mathrm{L}_{\xi}^{\omega_0} (\omega-\omega_0) - \iota_ {\xi}F_{\omega_0} + W_1^{-1}(\lambda e_n F_{\omega})+  \frac{1}{2}\Lambda \lambda e_n   e
\end{split}
\end{equation}
\begin{equation}\label{Q_boundary_part2}
\begin{split}
Qc &=  \frac{1}{2}[c,c] - \mathrm{L}_{\xi}^{\omega_0}c+\frac{1}{2}\iota_{\xi}\iota_{\xi}F_{\omega_0} - \left([c, \lambda e_n ]^{(a)}- \mathrm{L}_{\xi}^{\omega_0} (\lambda e_n)^{(a)}\right)(\omega - \omega_0)_a \\
Q\lambda &= [c, \lambda e_n ]^{(n)} - \mathrm{L}_{\xi}^{\omega_0} (\lambda e_n)^{(n)}\\
Q_{\xi} &=   [c, \lambda e_n ]^{(\bullet)} - \mathrm{L}_{\xi}^{\omega_0} (\lambda e_n)^{(\bullet)} - \frac{1}{2}[\xi,\xi]
\end{split}
\end{equation}
where $W_1^{-1}(\lambda e_n F_{\omega})$ is defined as in \eqref{e:ham_vf_L}.
The proof of $2\{S_0,S_1\}_f+\{S_1,S_1\}_g=0$ is a lengthy computation fully detailed in Appendix \ref{sec:appendix-BFV}.
\end{proof}
\begin{remark}
By setting $\lambda=0$, we can read the action of $Q$ on $c$ and $\xi$ as (a splitting by $\omega_0$) of the Atiyah algebroid structure on $TP/O(N-1,1)$ \cite{Mackenzie1987}, where $P$ is the orthogonal frame bundle of $M$ restricted to $\Sigma$.
\end{remark}
\subsection{Alternative variables}\label{s:altvar}
The $\xi$-dependent part of $S$ in \eqref{action_NC1} contains, in accordance with \eqref{constraints}, a repetition of the invariant constraint $ed_\omega e = 0$ which we have added to simplify the computations. This term may actually be removed by using the following symplectomorphism (cf. with \cite{CS2017}):
$$ c' = c + \iota_\xi (\omega-\omega_0) \qquad \xi^{'\dag}_a = \xi^{\dag}_a - (\omega - \omega_0)_a c^\dag$$
The resulting expressions of the action and symplectic form are:
\begin{align}\label{action_C1}
S= \int_{\Sigma} & c' e d_{\omega} e + \iota_{\xi} e e F_{\omega} + \lambda e_n \left(eF_\omega + \frac{1}{3!}\Lambda e^3\right) +\frac{1}{2} [c',c'] c^{\dag} - L^{\omega}_{\xi} c' c^{\dag}+ \frac{1}{2} \iota_{\xi}\iota_{\xi} F_{\omega}c^{\dag}\nonumber\\ &
 + [c', \lambda e_n ]^{(a)}\xi_a^{'\dag} + [c', \lambda e_n ]^{(n)}\lambda^\dag - L^{\omega}_{\xi} (\lambda e_n)^{(a)}\xi_a^{'\dag} \nonumber\\
 &- L^{\omega}_{\xi} (\lambda e_n)^{(n)}\lambda^\dag - \frac{1}{2}\iota_{[\xi,\xi]}\xi^{'\dag},\\
\label{symplectic_form_C1}
\varpi 
= \int_{\Sigma}& e \delta e \delta \omega + \delta c' \delta c^\dag + \delta \omega \delta(\iota_{\xi} c^\dag) + \delta \lambda \delta \lambda^\dag + \iota_{\delta \xi} \delta \xi^{\dag '}.
\end{align}

Note that the price for the simplication of the action is that \emph{primed} chart  is no longer a Darboux chart.
We can further transform \eqref{action_C1} and \eqref{symplectic_form_C1} in order to avoid using components. 
Since $\lambda^\dag$ and $\xi^{\dag '}$ both take value in $\wedge^4 \mathcal{V}$ we can write them in terms of the basis $(e_a, e_n)$:
\begin{align*}
\lambda^{\dag} &= {\lambda^{\dag}}^{(123n)}e_1 \wedge e_2 \wedge e_3 \wedge e_n; \\
\xi_a^{\dag'}dx^a &= {\xi_a^{\dag'}}^{(123n)}dx^ae_1 \wedge e_2 \wedge e_3 \wedge e_n,  \; a=1,2,3.
\end{align*}
Now define the following fields:
\begin{align*}
x_a^{a\dag}dx^a &:= {\xi_a^{\dag'}}^{(123n)}dx^a e_b \wedge e_c \wedge e_n \quad a,b,c\in \{1,2,3\},\quad b,c\neq a\\
l^\dag &:= {\lambda^{\dag}}^{(123n)}e_1 \wedge e_2 \wedge e_3 \qquad
y^\dag := l^\dag + \sum_{a=1}^3 (-1)^a x_a^{a\dag}.
\end{align*}
Multiplying $y^\dag$ by $e_a$ and $e_n$ gives back the original fields $\lambda^\dag$ and $\xi^{\dag '}$: $e_n y^\dag = -\lambda^\dag$, $e_a y^\dag =  -\xi_a^{\dag '}$.

Using these properties it is easy to show that we can express the action $S$ and the symplectic form $\varpi$ on the new space of fields given by the bundle
\begin{equation}\label{LoctrivF1_C3}
\mathcal{F'}\longrightarrow \Omega_{nd}^1(\Sigma, \mathcal{V}),
\end{equation}
with local trivialisation on an open $\mathcal{U}_{\Sigma} \subset \Omega_{nd}^1(\Sigma, \mathcal{V})$
\begin{align*}
\mathcal{F}\simeq \mathcal{U}_{\Sigma} \times \mathcal{A}^{red}(\Sigma) &\oplus \left(\Omega_{4, \partial}^{0,2}[1]\oplus \mathfrak{X}[1](\Sigma) \oplus C^\infty[1](\Sigma)\right)\\
&\oplus \Omega_{\partial}^{3,2}[-1] \oplus \Omega_{\partial}^{3,3}[-1] ,
\end{align*}
where all the fields are denoted as in Theorem \ref{thm:BFVaction} but $y^\dag \in \Omega_{\partial}^{3,3}[-1]$:
\begin{align}\label{action_C3}
S= \int_{\Sigma} & c' e d_{\omega} e + \iota_{\xi} e e F_{\omega} + \lambda e_n \left(eF_\omega + \frac{1}{3!}\Lambda e^3\right) +\frac{1}{2} [c',c'] c^{\dag} - L^{\omega}_{\xi} c' c^{\dag}+ \frac{1}{2} \iota_{\xi}\iota_{\xi} F_{\omega}c^{\dag}\nonumber\\ &
 -[c', \lambda e_n ]y^{\dag} + L^{\omega}_{\xi} (\lambda e_n)y^{\dag} + \frac{1}{2}\iota_{[\xi,\xi]}e y^{\dag},
\end{align}
\begin{align}\label{symplectic_form_C3}
\varpi = \int_{\Sigma} e \delta e \delta \omega + \delta c' \delta c^\dag + \delta \omega \delta (\iota_\xi c^\dag) - \delta \lambda e_n \delta y^\dag+\iota_{\delta \xi} \delta (e y^\dag).
\end{align}
It is a simple computation to show that this two form is actually nondegenerate.
\begin{remark}
Equation \eqref{action_C3} is again a covariant version of the BFV action functional. Moreover, it has the advantage of not including the implicit terms of \eqref{action_C1} and satisfies by construction the classical master equation. It hence provides a good starting point for the AKSZ construction performed in \cite{CCS2020b}.
\end{remark}

\section{Generalisation to \texorpdfstring{$\dim(M)>4$}{dim(M)>4}}\label{s:generalisation}
In this section we generalise the results of the previous sections to dimensions $N=\dim(M)>4$. The construction is substantially unchanged while a few details have to be fixed. We recall the main steps and adapt them to the generalisation. 
\subsection{Extension of the reduced phase space to higher dimensions}\label{s:RPSext}
The classical fields of the theory are as in the $N =4$ case:
a nondegenerate coframe $e \in \Omega^{1,1}_{\partial}$ restricted to the boundary and an equivalence class of connections $[\omega] \in \mathcal{A}^{red}(\Sigma)$ where $\mathcal{A}^{red}(\Sigma)$ is the quotient under $\omega  \sim \omega + v $ for $v$ such that $e^{N-3} \wedge v =0$.  The symplectic structure of the geometric phase space is given by
\begin{align}\label{classical-boundary-symplform_d4}
\varpi = \int_{\Sigma} e^{N-3} \delta e \delta [\omega].
\end{align}

Let now $e_n$ be a fixed section of $\mathcal{V}$ completing the image of $e: T\Sigma \rightarrow \mathcal{V}$ to a basis of $\mathcal{V}$. The structural constraint is 
\begin{equation}\label{omegareprfix_d4}
e_n e^{N-4} d_{\omega} e \in \Ima W_{N-3}^{\partial,(1,1)}.
\end{equation}

\begin{theorem}\label{thm:omegadecomposition_d4}
Suppose that the  boundary metric $g^{\partial}$ is nondegenerate. Given any $\widetilde{\omega} \in \Omega ( \Sigma,\wedge^2 \mathcal{V})$, there is a unique decomposition 
\begin{equation} \label{omegadecomp_d4}
\widetilde{\omega}= \omega +v
\end{equation}
with $\omega$ and $v$ satisfying 
\begin{equation}\label{omegareprfix2_d4}
e^{N-3} v=0 \quad \text{ and } \quad e_n e^{N-4}d_{\omega} e \in \Ima W_{N-3}^{\partial,(1,1)}.
\end{equation}
\end{theorem}
\begin{proof}
Let $\widetilde{\omega} \in \Omega ( \Sigma,\wedge^2 \mathcal{V})$. From Lemma \ref{lem:Omega2,2_d4} we deduce that there exist unique $\sigma \in \Omega(\Sigma, \mathcal{V})$ and $v \in \text{Ker} W_{N-3}^{\partial,(1,2)}$ such that 
\begin{align*}
e_n e^{N-4}d_{\widetilde{\omega}} e = e^{N-3} \sigma + e_n e^{N-4} [v,e].
\end{align*}
We define $\omega := \widetilde{\omega} - v $. Then $\omega$ and $v$ satisfy \eqref{omegadecomp} and \eqref{omegareprfix2}.

To prove uniqueness, suppose that $\widetilde{\omega}= \omega_1 + v_1 = \omega_2 +v_2$ with $e^{N-3}v_i =0$ and $e_n e^{N-4} d_{\omega_i} e \in \Ima W_{N-3}^{\partial,(1,1)}$ for $i=1,2$. Hence 
$$e_n e^{N-4} d_{\omega_1} e- e_n e^{N-4} d_{\omega_2} e = e_n e^{N-4} [v_2-v_1, e] \in  \Ima W_{N-3}^{\partial,(1,1)}.$$ Hence from Lemma \ref{lem:Omega2,1_d4} and \ref{lem:Omega2,2_d4}, we deduce $v_2-v_1 =0$, since $v_2-v_1 \in Ker W_{N-3}^{\partial,(1,2)}$.
\end{proof}
\begin{corollary}
The field $\omega$ in the decomposition \eqref{omegadecomp} depends only on the equivalence class $[\omega] \in \mathcal{A}^{red}(\Sigma)$.
\end{corollary}
\begin{proof}
Let $\widetilde{\omega}_1, \widetilde{\omega}_2 \in [\omega]$. Hence $\widetilde{\omega}_1- \widetilde{\omega}_2= \tilde{v} \in \text{Ker} W_{N-3}^{\partial,(1,2)}$. Applying Theorem \ref{thm:omegadecomposition_d4} we get $\omega_1$, $v_1$, $\omega_2$, $v_2$ such that $v_1, v_2 \in \text{Ker} W_{N-3}^{\partial,(1,2)}$ and 
\begin{align*}
\widetilde{\omega}_1= \omega_1 + v_1 \qquad & e_n {e^{N-4}}d_{\omega_1} e \in \Ima W_{N-3}^{\partial,(1,1)} \\
\widetilde{\omega}_2= \omega_2 + v_2 \qquad & e_n {e^{N-4}}d_{\omega_2} e \in \Ima W_{N-3}^{\partial,(1,1)}.
\end{align*}
Subtracting these equations we get $\omega_2- \omega_1 = v_1-v_2 - \tilde{v} \in \text{Ker} W_{N-3}^{\partial,(1,2)}$ and
$e_n {e^{N-4}}[\omega_1- \omega_2,e ] \in \Ima W_{N-3}^{\partial,(1,1)}$.  Hence, from \eqref{envenotinIme_d4}, we deduce $\omega_1 = \omega_2$.
\end{proof}

As for $N =4$, we consider the following constraints defined on $\mathcal{F}^{\partial}_{PC}$ using $\omega \in [\omega]$ defined in Theorem \ref{thm:omegadecomposition_d4}, hence satisfying \eqref{omegareprfix_d4}:
\begin{subequations}\label{constraints_d4}
\begin{equation}
L_c = \int_{\Sigma} c e^{N-3} d_{\omega} e 
\end{equation}
\begin{equation}
P_{\xi}= \int_{\Sigma}  \iota_{\xi} e e^{N-3} F_{\omega} + \iota_{\xi} (\omega-\omega_0) e^{N-3} d_{\omega} e
\end{equation}
\begin{equation}
H_{\lambda} = \int_{\Sigma} \lambda e_n \left( e^{N-3} F_{\omega} +\frac{1}{(N-1)!}\Lambda e^{N-1} \right),
\end{equation}
\end{subequations}
and Theorem \ref{thm:first-class-constraints} holds verbatim for these constraints too. 

\subsection{Extension of BFV data to higher dimensions}\label{s:BFVext}
Since the the brackets between the constraints are the same as in the $N =4$ case (Theorem \ref{thm:first-class-constraints}), the BFV action will have a similar expression too. For reference purposes, below we write the general version of Theorem \ref{thm:BFVaction}:
\begin{theorem}\label{thm:BFVaction_d4}
Let $g^\partial$ be nondegenerate on $\Sigma$. Let $\mathcal{F}$ be the bundle
\begin{equation}
\mathcal{F} \longrightarrow \Omega_{nd}^1(\Sigma, \mathcal{V}),
\end{equation}
with local trivialisation on an open $\mathcal{U}_{\Sigma} \subset \Omega_{nd}^1(\Sigma, \mathcal{V})$
\begin{equation}\label{LoctrivF1_d4}
\mathcal{F}\simeq \mathcal{U}_{\Sigma} \times \mathcal{A}^{red}(\Sigma) \oplus T^* \left(\Omega_{\partial}^{0,2}[1]\oplus \mathfrak{X}[1](\Sigma) \oplus \Omega_\partial^{0,0}[1]\right),
\end{equation}
and fields denoted by $e \in \mathcal{U}_{\Sigma}$ and $\omega \in \mathcal{A}^{red}(\Sigma)$ in degree zero, $c \in\Omega_{\partial}^{0,2}[1]$, $\xi \in\mathfrak{X}[1](\Sigma)$ and $\lambda\in \Omega^{0,0}[1]$ in degree one, $c^\dag\in\Omega_{\partial}^{N-1,N-2}[-1]$, $\lambda^\dag\in\Omega_{\partial}^{N-1,N}[-1]$ and $\xi^\dag\in\Omega^{1,0}[-1]\otimes\Omega_{\partial}^{N-1,N}$ in degree minus one, together with a fixed  $e_n \in \Gamma(\mathcal{V})$, completing the image of elements $e \in\mathcal{U}^{\Sigma}$ to a basis of  $\mathcal{V}$;
define a symplectic form and an action functional on $\mathcal{F}$ respectively by
\begin{align}
\varpi = \int_{\Sigma} &e^{N-3} \delta e \delta \omega + \delta c \delta c^\dag + \delta \lambda \delta \lambda^\dag + \iota_{\delta \xi} \delta \xi^\dag,\label{symplectic_form_NC1_d4} \\
S= \int_{\Sigma} & c e^{N-3} d_{\omega} e + \iota_{\xi} e e^{N-3} F_{\omega} + \iota_{\xi} (\omega-\omega_0) e^{N-3} d_{\omega} e +\lambda e_n  e^{N-3} F_{\omega} \nonumber\\ 
& +\frac{1}{(N-1)!}\Lambda\lambda e_n e^{N-1} +\frac{1}{2} [c,c] c^{\dag} - \mathrm{L}_{\xi}^{\omega_0} c c^{\dag}+\frac{1}{2}\iota_{\xi}\iota_{\xi}F_{\omega_0} c^{\dag}+ [c, \lambda e_n ]^{(n)}\lambda^\dag \nonumber\\
 &+ [c, \lambda e_n ]^{(a)}(\xi_a^{\dag}- (\omega - \omega_0)_a c^\dag)   - \mathrm{L}_{\xi}^{\omega_0} (\lambda e_n)^{(a)}(\xi_a^{\dag}- (\omega - \omega_0)_a c^\dag)\nonumber\\
 & - \mathrm{L}_{\xi}^{\omega_0} (\lambda e_n)^{(n)}\lambda^\dag - \frac{1}{2}\iota_{[\xi,\xi]}\xi^{\dag} \label{action_NC1_d4}
\end{align}
where $\omega$ satisfies the additional requirement $e_n e^{N-4} d_{\omega} e \in \Ima W_{N-3}^{\partial,(1,1)}$.
Then the triple $(\mathcal{F}, \varpi, S)$ forms a BFV structure on $\Sigma$.\end{theorem}
The covariant version of the action and symplectic form are
\begin{align}\label{action_C3_d4}
S= \int_{\Sigma} & c' e^{N-3} d_{\omega} e + \iota_{\xi} e e^{N-3} F_{\omega} +\lambda e_n \left( e^{N-3} F_{\omega} +\frac{1}{(N-1)!}\Lambda e^{N-1} \right) +\frac{1}{2} [c',c'] c^{\dag}\nonumber\\ 
& - L^{\omega}_{\xi} c' c^{\dag}+\frac{1}{2} \iota_{\xi}\iota_{\xi} F_{\omega}c^{\dag}+[c', \lambda e_n ]y^{\dag} - L^{\omega}_{\xi} (\lambda e_n)y^{\dag} - \frac{1}{2}\iota_{[\xi,\xi]}e y^{\dag},\\
\label{symplectic_form_C3_d4}
\varpi = \int_{\Sigma} &e^{N-3} \delta e \delta \omega + \delta c' \delta c^\dag - \delta \omega \delta (\iota_\xi c^\dag) + \delta \lambda e_n \delta y^\dag+\iota_{\delta \xi} \delta (e y^\dag).
\end{align}

\appendix
\section{Lengthy proofs of Section \ref{s:technical}}\label{sec:appendix-CME}
In this appendix we prove Lemmas \ref{lem:We_bulk}, \ref{lem:We_boundary} and \ref{lem:varrho12}.

\begin{proof}[Proof of Lemma \ref{lem:We_bulk}]
For each of the properties we use the following demonstrative scheme. We first compute the number of independent equations that a quantity must satisfy in order to lie in the kernel of the map under consideration. We then compare it to the dimension of the domain. If they agree, then the function is injective, otherwise comparing it with the dimension of the codomain we can deduce whether the map  is surjective. 
\begin{enumerate}
\item Consider $W_{N-3}^{ (2,1)}: \Omega^{2,1}  \longrightarrow \Omega^{N-1,N-2}:$

the dimension of the spaces are 
$\dim \Omega^{2,1} = \binom{N}{2}\binom{N}{1}$ and $\dim \Omega^{N-1,N-2} = \binom{N}{N-2}\binom{N}{N-1}$. Notice that the two dimensions agree.  The kernel of $W_{N-3}^{ (2,1)}$ is defined by the following set of equations:
\begin{align*}
X_{\mu_1 \mu_2}^a e_a \wedge e_{\mu_3} \wedge \dots \wedge e_{\mu_{N-1}} dx^{\mu_1}dx^{\mu_2}\dots dx^{\mu_{N-1}}=0
\end{align*}
where we used the vectors $e_a = e (\partial_a)$ as a basis for $\mathcal{V}$. Let now $1 \leq k \leq N$. Since $\{dx^{\mu_1}dx^{\mu_2}\dots dx^{\mu_{N-1}}\}$ is a basis for $\Omega^{N-1}(M)$ we obtain $N$ equations of the form
\begin{align*}
\sum_{\sigma} X_{\mu_{\sigma(1)} \mu_{\sigma(2)}}^a e_a \wedge e_{\mu_{\sigma(3)}} \wedge \dots \wedge e_{\mu_{\sigma(N-1)}} =0  
\end{align*}
where $\sigma$ runs on all permutations of $N-1$ elements and $1\leq \mu_i \leq N$, $\mu_i \neq k$ for all $1\leq i \leq N-1$. Recall now that $e_a \wedge e_{\mu_{\sigma(3)}} \wedge \dots \wedge e_{\mu_{\sigma(N-1)}}$ is a basis of $\wedge^{N-2}\mathcal{V}$. Hence we obtain the following equations:
\begin{align*}
 X_{ij}^k =0 \qquad & 1 \leq i,j\leq N \; i \neq j, \;  i,j \neq k \\
 \sum_{i\neq k, i\neq j} X_{ij}^i =0 \qquad & \forall 1 \leq j\leq N \; j \neq k 
\end{align*}
Letting now $k$ vary in $\{1, \dots , N\}$ we obtain the following equations:
\begin{align*}
 X_{ij}^k =0 \qquad & 1 \leq i,j,k\leq N \; i\neq j \neq k \neq i \\
 \sum_{i\neq k, i\neq j} X_{ij}^i =0 \qquad & \forall j \neq k,  \; 1 \leq k,j\leq N  .
\end{align*}
It is easy to check that these equations are independent.
The total number of equations defining the kernel is then $\frac{N(N-1)(N-2)}{2}+(N-1)N = \frac{(N-1)N^2}{2}$ which coincides with both the dimensions of the domain and codomain. Hence  $W_{N-3}^{ (2,1)}$ is bijective.
\item $ W_{N-3}^{ (2,2)}: \Omega^{2,2}  \rightarrow \Omega^{N-1,N-1}$ cannot be injective since $\Omega^{2,2}$ has $\frac{N^2(N-1)^2}{4}$ degrees of freedom, while $\Omega^{N-1,N-1}$ has just $N^2$ degrees of freedom and $N \geq 4$.
\end{enumerate}
\end{proof}
\begin{proof}[Proof of Lemma \ref{lem:We_boundary}]
For each of the properties we use the same scheme of the proof of Lemma \ref{lem:We_bulk}.
\begin{enumerate}
\item The proof of $W_{N-3}^{\partial, (2,1)}$ is analogous to that of $W_{N-3}^{ (2,1)}$ with the difference that now $k$ is fixed to be the transversal direction (conventionally $k=N$). Hence we get the following set of equations:
\begin{align*}
 X_{ij}^N =0 \qquad & 1 \leq i,j\leq N-1 \; i \neq j \\
 \sum_{i\neq j} X_{ij}^i =0 \qquad & \forall 1 \leq j\leq N-1 
\end{align*}
which are $\frac{(N-1)(N-2)}{2}+(N-1) = \frac{N(N-1)}{2}$ which is exactly the number of degrees of freedom of   $\Omega_{ \partial}^{N-1,N-2}$. Hence $W_{N-3}^{\partial, (2,1)}$ is surjective but not injective. In particular $\dim \text{Ker} W_{N-3}^{\partial, (2,1)}= \frac{N(N-1)(N-2)}{2}- \frac{N(N-1)}{2}= \frac{N(N-1)}{2}(N-3)$.
\item Consider $W_{N-3}^{ \partial, (1,1)}: \Omega_{ \partial}^{1,1}  \longrightarrow \Omega_{ \partial}^{N-2,N-2}$:
the dimension of the spaces are 
$\dim \Omega_{ \partial}^{1,1} = (N-1)N$ and $\dim \Omega_{ \partial}^{N-2,N-2} = (N-1)\frac{N(N-1)}{2}$. The kernel of $W_{N-3}^{ \partial, (1,1)}$ is defined by the following set of equations:
\begin{align*}
X_{\mu_1}^a e_a \wedge e_{\mu_2} \wedge \dots \wedge e_{\mu_{N-2}} dx^{\mu_1}dx^{\mu_2}\dots dx^{\mu_{N-2}}=0
\end{align*}
where we used $e_a$ as a basis for $\mathcal{V}$. Let now $k=N$ be the transversal direction and let $k'\in\{1, \dots N-1\}$. Since $\{dx^{\mu_1}dx^{\mu_2}\dots dx^{\mu_{N-2}}\}$ is a basis for $\Omega^{N-2}(M)$ we obtain $N-1$ equations of the form
\begin{align*}
\sum_{\sigma} X_{\mu_{\sigma(1)}}^a e_a \wedge e_{\mu_{\sigma(2)}} \wedge \dots \wedge e_{\mu_{\sigma(N-2)}} =0  
\end{align*}
where $\sigma$ runs on all permutations of $N-2$ elements and $1\leq \mu_i \leq N-1$, $\mu_i \neq k'$ for all $1\leq i \leq N-2$. Recall now that $e_a \wedge e_{\mu_{\sigma(2)}} \wedge \dots \wedge e_{\mu_{\sigma(N-2)}}$ is a basis of $\wedge^{N-2}\mathcal{V}$. Hence we obtain the following equations:
\begin{align*}
 X_{i}^k =0 \qquad & 1 \leq i \leq N-1 \; i \neq k' \\
 X_{i}^{k'} =0 \qquad & 1 \leq i \leq N-1 \; i \neq k'\\
 \sum_{i\neq k, i\neq k'} X_{i}^i =0 \qquad & 
\end{align*}
Letting now $k'$ vary in $\{1, \dots , N-1\}$ we obtain the following equations:
\begin{align*}
 X_{i}^k =0 \qquad & 1 \leq i \leq N-1 \\
 X_{i}^{j} =0 \qquad & 1 \leq i,j \leq N-1 \; i \neq j\\
 \sum_{i\neq k, i\neq j} X_{i}^i =0 \qquad &  1 \leq j \leq N-1
\end{align*}
It is easy to check that these equations are independent.
The total number of equations defining the kernel is then $(N-1)+(N-1)(N-2)+(N-1) = (N-1)N$ which coincides with number of degrees of freedom of the domain. Hence  $W_{N-3}^{ \partial, (1,1)}$ is injective but not surjective.
\item Consider $W_{N-3}^{ \partial, (1,2)}: \Omega_{ \partial}^{1,2}  \longrightarrow \Omega_{ \partial}^{N-2,N-1}$:
the dimensions of domain and codomain are 
$\dim \Omega_{ \partial}^{1,2} = (N-1)\frac{N(N-1)}{2}$ and $\dim \Omega_{ \partial}^{N-2,N-1} = (N-1)N$. The kernel of $W_{N-3}^{ \partial, (1,2)}$ is defined by the following set of equations:
\begin{align*}
X_{\mu_1}^{ab} e_a e_b \wedge e_{\mu_2} \wedge \dots \wedge e_{\mu_{N-2}} dx^{\mu_1}dx^{\mu_2}\dots dx^{\mu_{N-2}}=0
\end{align*}
where we used $e_a$ as a basis for $\mathcal{V}$. Let now $k=N$ be the transversal direction and let $k'\in\{1, \dots N-1\}$. Since $\{dx^{\mu_1}dx^{\mu_2}\dots dx^{\mu_{N-2}}\}$ is a basis for $\Omega^{N-2}(M)$ we obtain $N-1$ equations of the form
\begin{align*}
\sum_{\sigma} X_{\mu_{\sigma(1)}}^{ab} e_a e_b \wedge e_{\mu_{\sigma(2)}} \wedge \dots \wedge e_{\mu_{\sigma(N-2)}} =0  
\end{align*}
where $\sigma$ runs on all permutations of $N-2$ elements and $1\leq \mu_i \leq N-1$, $\mu_i \neq k'$ for all $1\leq i \leq N-2$. Recall now that $e_a e_b \wedge e_{\mu_{\sigma(2)}} \wedge \dots \wedge e_{\mu_{\sigma(N-2)}}$ is a basis of $\wedge^{N-1}\mathcal{V}$. Hence we obtain the following equations:
\begin{align*}
 X_{i}^{N k'} =0 \qquad & 1 \leq i \leq N-1 \; i \neq k' \\
 \sum_{i\neq N, i\neq k'} X_{i}^{iN} =0 \qquad & 
  \sum_{i\neq N, i\neq k'} X_{i}^{ik'} =0 
\end{align*}
Letting now $k'$ vary in $\{1, \dots , N-1\}$ we obtain the following equations:
\begin{subequations}\label{e:conditions_kernel_v}
\begin{align}
 X_{i}^{N j} =0 \qquad & 1 \leq i,j \leq N-1 \; i \neq j \\
 \sum_{i\neq N, i\neq j} X_{i}^{iN} =0 \qquad & 1 \leq j \leq N-1 \\
  \sum_{i\neq N, i\neq j} X_{i}^{ij} =0 \qquad & 1 \leq j \leq N-1 
\end{align}
\end{subequations}
It is easy to check that these equations are independent.
The total number of equations defining the kernel is then $(N-1)+(N-1)(N-2)+(N-1) = (N-1)N$ which coincides with number of degrees of freedom of the codomain. Hence  $W_{N-3}^{ \partial, (1,2)}$ is surjective but not injective. In particular $\dim \text{Ker} W_{N-3}^{\partial, (1,2)}= (N-1)\frac{N(N-1)}{2}- N(N-1)= \frac{N(N-1)}{2}(N-3)$.
\item Is a direct consequence of the previous parts.
\item Consider $W_{N-4}^{ \partial, (2,1)}: \Omega_{ \partial}^{2,1}  \longrightarrow \Omega_{ \partial}^{N-2,N-3}$:
the dimension of domain and codomain are 
$\dim \Omega_{ \partial}^{2,1} = \frac{(N-2)(N-1)}{2}N$ and $\dim \Omega_{ \partial}^{N-2,N-3} = (N-1)\frac{N(N-1)(N-2)}{6}$. The kernel of $W_{N-4}^{ \partial, (2,1)}$ is defined by the following set of equations:
\begin{align*}
X_{\mu_1 \mu_2}^a e_a \wedge e_{\mu_3} \wedge \dots \wedge e_{\mu_{N-2}} dx^{\mu_1}dx^{\mu_2}\dots dx^{\mu_{N-2}}=0
\end{align*}
where we used $e_a$ as a basis for $\mathcal{V}$. Let now $k=N$ be the transversal direction and let $k'\in\{1, \dots N-1\}$. Since $\{dx^{\mu_1}dx^{\mu_2}\dots dx^{\mu_{N-2}}\}$ is a basis for $\Omega^{N-2}(M)$ we obtain $N-1$ equations of the form
\begin{align*}
\sum_{\sigma} X_{\mu_{\sigma(1)}\mu_{\sigma(2)}}^a e_a \wedge e_{\mu_{\sigma(3)}} \wedge \dots \wedge e_{\mu_{\sigma(N-2)}} =0  
\end{align*}
where $\sigma$ runs on all permutations of $N-2$ elements and $1\leq \mu_i \leq N-1$, $\mu_i \neq k'$ for all $1\leq i \leq N-2$. Recall now that $e_a \wedge e_{\mu_{\sigma(3)}} \wedge \dots \wedge e_{\mu_{\sigma(N-2)}}$ is a basis of $\wedge^{N-3}\mathcal{V}$. Hence we obtain the following equations:
\begin{align*}
 X_{ij}^N =0 \qquad & 1 \leq i,j \leq N-1 \; i,j \neq k' \\
 X_{ij}^{k'} =0 \qquad & 1 \leq i,j \leq N-1 \; i,j \neq k'\\
 \sum_{i\neq N, i\neq k'} X_{ij}^i =0 \qquad &  1 \leq j \leq N-1 \; j \neq k'
\end{align*}
Letting now $k'$ vary in $\{1, \dots , N-1\}$ we obtain the following equations:
\begin{align*}
 X_{ij}^N =0 \qquad & 1 \leq i,j \leq N-1 \\
 X_{ij}^{j'} =0 \qquad & 1 \leq i,j,j' \leq N-1 \; i,j \neq j' \; i \neq j\\
 \sum_{i\neq k, i\neq j'} X_{ij}^i =0 \qquad &  1 \leq j, j' \leq N-1 \; j \neq j'
\end{align*}
It is easy to check that these equations are independent.
The total number of equations defining the kernel is then $\frac{(N-2)(N-1)}{2}+\frac{(N-3)(N-2)(N-1)}{2} (N-2)(N-1)= \frac{(N-2)(N-1)N}{2}$ which coincides with number of degrees of freedom of the domain. Hence  $W_{N-4}^{ \partial, (2,1)}$ is injective but not surjective.
\end{enumerate}
\end{proof}

\begin{proof}[Proof of Lemma \ref{lem:varrho12}]
Consider $\varrho |_{\text{Ker} W_{N-3}^{\partial, (1,2)}}: \text{Ker} W_{N-3}^{\partial, (1,2)} \rightarrow \Omega_{ \partial}^{2,1}$. From \hyperref[lem:kerWe12]{\ref*{lem:We_boundary}.(\ref*{lem:kerWe12})} we know that $\dim \text{Ker} W_{N-3}^{\partial, (2,1)}= \frac{N(N-1)}{2}(N-3)$. An element $v \in  \text{Ker} W_{N-3}^{\partial, (1,2)}$ must satisfy equations \eqref{e:conditions_kernel_v}.
The kernel of $\varrho$ is defined by the following set of equations:
\footnote{Here we use that in every point we can find a basis in $\mathcal{V}$ such that $e_\mu^i= \delta_\mu^i$: $[v,e]_{\mu_1\mu_2}^a= v_{\mu_1}^{ab}\eta{bc}e_{\mu_2}^c- v_{\mu_2}^{ab}\eta{bc}e_{\mu_1}^c= v_{\mu_1}^{ab}e_{b}^d\eta{dc}e_{\mu_2}^c- v_{\mu_2}^{ab}e_{b}^d\eta{dc}e_{\mu_1}^c$}
\begin{align*}
[v,e]_{\mu_1\mu_2}^a= v_{\mu_1}^{ab} g^{\partial}_{b \mu_2} - v_{\mu_2}^{ab}g^{\partial}_{b\mu_1}=0.
\end{align*}
Using now normal geodesic coordinates, we can diagonalise $g^\partial$ with eigenvalues on the diagonal $\alpha_{\mu} \in \{1,-1,0\}$:
\begin{align*}
[v,e]_{\mu_1\mu_2}^a= v_{\mu_1}^{a\mu_2} \alpha_{ \mu_2} - v_{\mu_2}^{a\mu_1} \alpha_{\mu_1}=0
\end{align*}
If $g^\partial$ is nondegenerate these equations become $v_{\mu_1}^{a\mu_2} =\pm v_{\mu_2}^{a\mu_1}$. Namely, using $v \in  \text{Ker} W_{N-3}^{\partial, (2,1)}$ we get
\begin{align*}
v^{ij}_i =0 \qquad & 0 \leq i,j\leq N-1, \; i\neq j\\
v_{i_1}^{i_2 i_3}= \pm v_{i_2}^{i_1 i_3} \qquad & 0 \leq i_1, i_2, i_3 \leq N-1 \; i_1,i_2 \neq i_3, \; i_2 \neq i_1 
\end{align*}
It is easy to check that these equations are independent.
The total number of equations defining the kernel is then $(N-1)(N-3) + \frac{(N-1)(N-2)(N-3)}{2}= \frac{N(N-1)}{2}(N-3)$ which coincides with number of degrees of freedom of the domain. Hence  $\varrho |_{\text{Ker} W_{N-3}^{\partial, (1,2)}} $ is injective.
\end{proof}

\section[Lengthy proofs]{Lengthy proofs of Section \ref{s:BFV}}\label{sec:appendix-BFV}
We complete the proof of Theorem \ref{thm:BFVaction}. Namely we prove here explicitly that $2\{S_0,S_1\}_f+\{S_1,S_1\}_g=0$. From the expression of $Q$ and  $\{S_0,S_1\}_f= \iota_{Q_0}\iota_{Q_1}\varpi_f$, we get:
\begin{align}\label{S0S1f}
    &\{S_0,S_1\}_f  \\
    = &-\unl{ [c, \lambda e_n ]^{(b)}([c,e])_b^{(a)}(\xi_a^{\dag}- (\omega - \omega_0)_a c^\dag) }{S0S1f1}
    -\unl{ [c, \lambda e_n ]^{(b)}([c,e])_b^{(n)}\lambda^\dag }{S0S1f2}\nonumber \\   
    &+\unl{L_{\xi}^{\omega_0} (\lambda e_n)^{(b)}([c,e])_b^{(a)}(\xi_a^{\dag}- (\omega - \omega_0)_a c^\dag)}{S0S1f3} 
    +\unl{L_{\xi}^{\omega_0} (\lambda e_n)^{(b)}([c,e])_b^{(n)}\lambda^\dag }{S0S1f4}\nonumber \\
    & +\unl{ [c, \lambda e_n ]^{(b)}(L_{\xi}^{\omega_0} e)_b^{(a)}(\xi_a^{\dag}- (\omega - \omega_0)_a c^\dag) }{S0S1f5}
    +\unl{ [c, \lambda e_n ]^{(b)}(L_{\xi}^{\omega_0} e)_b^{(n)}\lambda^\dag }{S0S1f6} \nonumber \\
    &-\unl{L_{\xi}^{\omega_0} (\lambda e_n)^{(b)}(L_{\xi}^{\omega_0} e)_b^{(a)}(\xi_a^{\dag}- (\omega - \omega_0)_a c^\dag)}{S0S1f7}
    -\unl{L_{\xi}^{\omega_0} (\lambda e_n)^{(b)}(L_{\xi}^{\omega_0} e)_b^{(n)}\lambda^\dag }{S0S1f8}\nonumber \\
    & -\unl{ [c, \lambda e_n ]^{(b)}(d_{\omega}(\lambda e_n))_b^{(a)}(\xi_a^{\dag}- (\omega - \omega_0)_a c^\dag)}{S0S1f9}
    -\unl{ [c, \lambda e_n ]^{(b)}(d_{\omega}(\lambda e_n))_b^{(n)}\lambda^\dag }{S0S1f10}\nonumber \\
    &+\unl{L_{\xi}^{\omega_0} (\lambda e_n)^{(b)}(d_{\omega}(\lambda e_n))_b^{(a)}(\xi_a^{\dag}- (\omega - \omega_0)_a c^\dag)}{S0S1f11}
    +\unl{L_{\xi}^{\omega_0} (\lambda e_n)^{(b)}(d_{\omega}(\lambda e_n))_b^{(n)}\lambda^\dag }{S0S1f12}\nonumber \\
    &-\unl{ [c, \lambda e_n ]^{(b)}( \lambda \sigma)_b^{(a)}(\xi_a^{\dag}- (\omega - \omega_0)_a c^\dag) }{S0S1f13}
    -\unl{ [c, \lambda e_n ]^{(b)}( \lambda \sigma)_b^{(n)}\lambda^\dag }{S0S1f14}\nonumber \\
    & +\unl{L_{\xi}^{\omega_0} (\lambda e_n)^{(b)}( \lambda \sigma)_b^{(a)}(\xi_a^{\dag}- (\omega - \omega_0)_a c^\dag)}{S0S1f15}
    +\unl{L_{\xi}^{\omega_0} (\lambda e_n)^{(b)}( \lambda \sigma)_b^{(n)}\lambda^\dag }{S0S1f16}\nonumber \\
    &-\unl{[c, \lambda e_n ]^{(a)}(d_{\omega}c)_a c^\dag}{S0S1f17}
    +\unl{ L_{\xi}^{\omega_0} (\lambda e_n)^{(a)} (d_{\omega}c)_a c^\dag }{S0S1f18}
    +\unl{[c, \lambda e_n ]^{(a)}(L_{\xi}^{\omega_0} (\omega-\omega_0))_a c^\dag}{S0S1f19}\nonumber \\
    &-\unl{ L_{\xi}^{\omega_0} (\lambda e_n)^{(a)} (L_{\xi}^{\omega_0} (\omega-\omega_0))_a c^\dag}{S0S1f20}
    -\unl{[c, \lambda e_n ]^{(a)}( W_1^{-1}(\lambda e_n F_{\omega}))_a c^\dag}{S0S1f21}\nonumber \\
    &+\unl{ L_{\xi}^{\omega_0} (\lambda e_n)^{(a)} ( W_1^{-1}(\lambda e_n F_{\omega}))_a c^\dag }{S0S1f22}
    +\unl{[c, \lambda e_n ]^{(a)}(\iota_{\xi}F_{\omega_0})_a c^\dag}{S0S1f23}\nonumber \\
    &-\unl{ L_{\xi}^{\omega_0} (\lambda e_n)^{(a)} (\iota_{\xi}F_{\omega_0})_a c^\dag}{S0S1f24}
    -\unl{\frac{1}{2}\Lambda[c, \lambda e_n ]^{(a)}e_a \lambda e_n c^\dag}{S0S1f25}
    +\unl{\frac{1}{2}\Lambda L_{\xi}^{\omega_0} (\lambda e_n)^{(a)}e_a \lambda e_n c^\dag .}{S0S1f26} \nonumber
\end{align}
From $\{S_1,S_1\}_g= \iota_{Q_1}\iota_{Q_1}\varpi_g$ we get:
\begin{align}\label{S1S1g}
    &\frac{1}{2}\{S_1,S_1\}_g  \\
  &=\unl{\frac{1}{2}[[c,c],c] c^{\dag}}{S1S1g1}
    -\unl{\frac{1}{2}[c,c] L_{\xi}^{\omega_0} c^{\dag}}{S1S1g2}
    +\unl{\frac{1}{2} [[c,c], \lambda e_n ]^{(a)}(\xi_a^{\dag}- (\omega - \omega_0)_a c^\dag) }{S1S1g3}
    +\unl{ \frac{1}{2}[[c,c], \lambda e_n ]^{(n)}\lambda^\dag}{S1S1g4} \nonumber \\
    &-\unl{[L_{\xi}^{\omega_0}c,c] c^{\dag} }{S1S1g5}
    +\unl{L_{\xi}^{\omega_0}c L_{\xi}^{\omega_0} c^{\dag}}{S1S1g6}
    -\unl{ [L_{\xi}^{\omega_0}c, \lambda e_n ]^{(a)}(\xi_a^{\dag}+ (\omega - \omega_0)_a c^\dag) }{S1S1g7} \nonumber \\
    &-\unl{ [L_{\xi}^{\omega_0}c, \lambda e_n ]^{(n)}\lambda^\dag }{S1S1g8}
    -\unl{[[c, \lambda e_n ]^{(a)}(\omega - \omega_0)_a,c] c^{\dag} }{S1S1g9}
    +\unl{[c, \lambda e_n ]^{(a)}(\omega -\omega_0)_a L_{\xi}^{\omega_0} c^{\dag}}{S1S1g10} \nonumber \\ 
    &-\unl{[[c, \lambda e_n ]^{(a)}(\omega - \omega_0)_a, \lambda e_n ]^{(a)}(\xi_a^{\dag}+ (\omega - \omega_0)_a c^\dag) }{S1S1g11}
    -\unl{ [[c, \lambda e_n ]^{(a)}(\omega - \omega_0)_a, \lambda e_n ]^{(n)}\lambda^\dag }{S1S1g12}\nonumber \\
    &+\unl{[L_{\xi}^{\omega_0} (\lambda e_n)^{(a)}(\omega - \omega_0)_a,c] c^{\dag} }{S1S1g13}
    -\unl{L_{\xi}^{\omega_0} (\lambda e_n)^{(a)}(\omega - \omega_0)_a L_{\xi}^{\omega_0} c^{\dag}}{S1S1g14}\nonumber \\
    &+\unl{ [L_{\xi}^{\omega_0} (\lambda e_n)^{(a)}(\omega - \omega_0)_a, \lambda e_n ]^{(a)}(\xi_a^{\dag}- (\omega - \omega_0)_a c^\dag) }{S1S1g15}
    +\unl{ [L_{\xi}^{\omega_0} (\lambda e_n)^{(a)}(\omega - \omega_0)_a, \lambda e_n ]^{(n)}\lambda^\dag }{S1S1g16}\nonumber \\ 
    &-\unl{ [c, [c, \lambda e_n ]^{(n)} e_n ]^{(a)}(\xi_a^{\dag} - (\omega - \omega_0)_a c^\dag)}{S1S1g17}
    -\unl{ [c, [c, \lambda e_n ]^{(n)} e_n ]^{(n)}\lambda^\dag}{S1S1g18}\nonumber \\
    &+\unl{ L_{\xi}^{\omega_0} ([c, \lambda e_n ]^{(n)} e_n)^{(a)}(\xi_a^{\dag}- (\omega - \omega_0)_a c^\dag)}{S1S1g19}
    +\unl{ L_{\xi}^{\omega_0} ([c, \lambda e_n ]^{(n)} e_n)^{(n)}\lambda^\dag}{S1S1g20}\nonumber \\ 
    &+\unl{ [c, L_{\xi}^{\omega_0} (\lambda e_n)^{(n)} e_n ]^{(a)}(\xi_a^{\dag} - (\omega - \omega_0)_a c^\dag) }{S1S1g21} 
    +\unl{ [c, L_{\xi}^{\omega_0} (\lambda e_n)^{(n)} e_n ]^{(n)}\lambda^\dag}{S1S1g22}\nonumber \\\intertext{}
   & -\unl{ L_{\xi}^{\omega_0} (L_{\xi}^{\omega_0} (\lambda e_n)^{(n)} e_n)^{(a)}(\xi_a^{\dag}- (\omega - \omega_0)_a c^\dag) }{S1S1g23}
    -\unl{ L_{\xi}^{\omega_0} (L_{\xi}^{\omega_0} (\lambda e_n)^{(n)} e_n)^{(n)}\lambda^\dag}{S1S1g24}\nonumber \\
    &-\unl{ [c, \lambda e_n ]^{(a)} {d_{\omega_0}}_a c c^{\dag} }{S1S1g25}
    -\unl{ ([c, \lambda e_n ]^{(b)} {d_{\omega_0}}_b(\lambda e_n))^{(a)}(\xi_a^{\dag}- (\omega - \omega_0)_a c^\dag) }{S1S1g26}\nonumber \\
    &-\unl{ ([c, \lambda e_n ]^{(a)} {d_{\omega_0}}_a(\lambda e_n))^{(n)}\lambda^\dag }{S1S1g27}
    -\unl{ [c, \lambda e_n ]^{(a)}(\partial_a \xi^b) \xi_b^{\dag}}{S1S1g28}
    -\unl{ [c, \lambda e_n ]^{(a)}\partial_b (\xi^b \xi_a^{\dag})}{S1S1g29}\nonumber \\
    &+\unl{ L_{\xi}^{\omega_0} (\lambda e_n)^{(a)} {d_{\omega_0}}_a c c^{\dag} }{S1S1g30}
    +\unl{ (L_{\xi}^{\omega_0} (\lambda e_n)^{(b)} {d_{\omega_0}}_b(\lambda e_n))^{(a)}(\xi_a^{\dag}- (\omega - \omega_0)_a c^\dag)}{S1S1g31}\nonumber \\
    &+\unl{ (L_{\xi}^{\omega_0} (\lambda e_n)^{(a)} {d_{\omega_0}}_a(\lambda e_n))^{(n)}\lambda^\dag }{S1S1g32}
    +\unl{ L_{\xi}^{\omega_0} (\lambda e_n)^{(a)}(\partial_a \xi^b) \xi_b^{\dag}}{S1S1g33}
    +\unl{ L_{\xi}^{\omega_0} (\lambda e_n)^{(a)}\partial_b (\xi^b \xi_a^{\dag})}{S1S1g34}\nonumber \\
    & +\unl{\frac{1}{2} \iota_{[\xi,\xi]} {d_{\omega_0}} c c^{\dag} }{S1S1g35}
    +\unl{\frac{1}{2}(\iota_{[\xi,\xi]} {d_{\omega_0}}(\lambda e_n))^{(a)}(\xi_a^{\dag}- (\omega - \omega_0)_a c^\dag) }{S1S1g36} 
    +\unl{\frac{1}{2} (\iota_{[\xi,\xi]} {d_{\omega_0}}(\lambda e_n))^{(n)}\lambda^\dag }{S1S1g37}\nonumber \\
    & +\unl{\frac{1}{2}[\xi,\xi]^a(\partial_a \xi^b) \xi_b^{\dag}}{S1S1g38}
    +\unl{\frac{1}{2} [\xi,\xi]^a\partial_b (\xi^b \xi_a^{\dag})}{S1S1g39}  +\unl{\frac{1}{2}[\iota_{\xi}\iota_{\xi}F_{\omega_0},c] c^{\dag}}{S1S1g40}
    -\unl{\frac{1}{2}\iota_{\xi}\iota_{\xi}F_{\omega_0} L_{\xi}^{\omega_0} c^{\dag}}{S1S1g41}\nonumber \\
    & +\unl{\frac{1}{2} [\iota_{\xi}\iota_{\xi}F_{\omega_0}, \lambda e_n ]^{(a)}(\xi_a^{\dag}- (\omega - \omega_0)_a c^\dag) }{S1S1g42}
    +\unl{ \frac{1}{2}[\iota_{\xi}\iota_{\xi}F_{\omega_0}, \lambda e_n ]^{(n)}\lambda^\dag}{S1S1g43}\nonumber \\
    & -\unl{[c, \lambda e_n ]^{(a)}(\iota_{\xi}F_{\omega_0})_a c^{\dag} }{S1S1g44}
    +\unl{ L_{\xi}^{\omega_0} (\lambda e_n)^{(a)}(\iota_{\xi}F_{\omega_0})_a c^{\dag} }{S1S1g45}
    +\unl{ \frac{1}{2}\iota_{[\xi,\xi]}\iota_{\xi}F_{\omega_0} c^{\dag}.}{S1S1g46} \nonumber
\end{align}
We now check term by term that the sum $2\{S_0,S_1\}_f+\{S_1,S_1\}_g$ is zero. We have: 
\begin{itemize}
\item \reft{S1S1g}{S1S1g1} =0 using (graded) Jacobi identity.
\item \reft{S1S1g}{S1S1g2} and \reft{S1S1g}{S1S1g5}: $$ -\frac{1}{2}L_{\xi}^{\omega_0}([c,c]  c^{\dag})=  -\frac{1}{2}[c,c] L_{\xi}^{\omega_0} c^{\dag} -[L_{\xi}^{\omega_0}c,c]  c^{\dag}.$$
\item 

\reft{S1S1g}{S1S1g3}, \reft{S1S1g}{S1S1g17}, \reft{S0S1f}{S0S1f1}: using (graded) Jacobi identity:
\begin{align*}
\frac{1}{2} [[c,c], \lambda e_n ]^{(a)}(\xi_a^{\dag}- (\omega - \omega_0)_a c^\dag)&- [c, [c, \lambda e_n ]^{(n)} e_n ]^{(a)}(\xi_a^{\dag} - (\omega - \omega_0)_a c^\dag)\\
&- [c, \lambda e_n ]^{(b)}([c,e])_b^{(a)}(\xi_a^{\dag}- (\omega - \omega_0)_a c^\dag)=0 
\end{align*}

\item 
\reft{S1S1g}{S1S1g4}, \reft{S1S1g}{S1S1g18}, \reft{S0S1f}{S0S1f2}: as before.
\item 
\reft{S1S1g}{S1S1g6}, \reft{S1S1g}{S1S1g35} and \reft{S1S1g}{S1S1g40} :
$$L_{\xi}^{\omega_0}c L_{\xi}^{\omega_0} c^{\dag}+\frac{1}{2} \iota_{[\xi,\xi]} {d_{\omega_0}} c c^{\dag}+\frac{1}{2}[\iota_{\xi}\iota_{\xi}F_{\omega_0},c] c^{\dag}= - d_{\omega_0} (\iota_{\xi} {d_{\omega_0}} c \iota_\xi c^\dag).$$
\item 
\reft{S1S1g}{S1S1g7}, \reft{S1S1g}{S1S1g21} and \reft{S0S1f}{S0S1f3}:
\begin{align*}
 -[L_{\xi}^{\omega_0}c, \lambda e_n ]^{(a)}(\xi_a^{\dag}+ (\omega - \omega_0)_a c^\dag)&+ [c, L_{\xi}^{\omega_0} (\lambda e_n)^{(n)} e_n ]^{(a)}(\xi_a^{\dag}  - (\omega - \omega_0)_a c^\dag)\\
&+L_{\xi}^{\omega_0} (\lambda e_n)^{(b)}([c,e])_b^{(a)}(\xi_a^{\dag}- (\omega - \omega_0)_a c^\dag) \\
=& -(L_{\xi}^{\omega_0}[c, \lambda e_n ])^{(a)}(\xi_a^{\dag}+ (\omega - \omega_0)_a c^\dag)
\end{align*}
We have $(L_{\xi}^{\omega_0}\omega)_a= L_{\xi}^{\omega_0}(\omega - \omega_0)_a + \partial_a \xi^c \omega_c$ and
$(L_{\xi}^{\omega_0}e)_a= L_{\xi}^{\omega_0}e_a + \partial_a \xi^c e_c$. 
\reft{S0S1f}{S0S1f5}, \reft{S1S1g}{S1S1g28}, \reft{S0S1f}{S0S1f19}:
\begin{align*}
&[c, \lambda e_n ]^{(b)}(L_{\xi}^{\omega_0} e)_b^{(a)}(\xi_a^{\dag}- (\omega - \omega_0)_a c^\dag)- [c, \lambda e_n ]^{(a)}(\partial_a \xi^b) \xi_b^{\dag}+[c, \lambda e_n ]^{(a)}(L_{\xi}^{\omega_0} \omega)_a c^\dag\\
=& [c, \lambda e_n ]^{(b)}(L_{\xi}^{\omega_0} e_b)^{(a)}(\xi_a^{\dag}- (\omega - \omega_0)_a c^\dag)+[c, \lambda e_n ]^{(b)}( \partial_b \xi^c e_c)^{(a)}(\xi_a^{\dag}- (\omega - \omega_0)_a c^\dag)\\
&-[c, \lambda e_n ]^{(b)}(\partial_b \xi^c) \xi_c^{\dag}+[c, \lambda e_n ]^{(a)}(L_{\xi}^{\omega_0} (\omega - \omega_0)_a) c^\dag+[c, \lambda e_n ]^{(a)}( \partial_a \xi^c \omega_c) c^\dag\\
=&  [c, \lambda e_n ]^{(b)}(L_{\xi}^{\omega_0} e_b)^{(a)}(\xi_a^{\dag}- (\omega - \omega_0)_a c^\dag)+[c, \lambda e_n ]^{(a)}(L_{\xi}^{\omega_0} (\omega - \omega_0)_a) c^\dag\\
\end{align*} 
\reft{S1S1g}{S1S1g29}: $- [c, \lambda e_n ]^{(a)}\partial_b (\xi^b \xi_a^{\dag})= \partial_b ([c, \lambda e_n ]^{(a)}\xi^b \xi_a^{\dag}) + L_{\xi}^{\omega_0}[c, \lambda e_n ]^{(a)} \xi_a^{\dag}$\\
\reft{S1S1g}{S1S1g19}, \reft{S1S1g}{S1S1g10} and previous relations:
\begin{align*}
&[c, \lambda e_n ]^{(b)}(L_{\xi}^{\omega_0} e_b)^{(a)}(\xi_a^{\dag}- (\omega - \omega_0)_a c^\dag)+[c, \lambda e_n ]^{(a)}(L_{\xi}^{\omega_0} (\omega - \omega_0)_a) c^\dag\\
&+ L_{\xi}^{\omega_0} ([c, \lambda e_n ]^{(n)} e_n)^{(a)}(\xi_a^{\dag}- (\omega - \omega_0)_a c^\dag)+ L_{\xi}^{\omega_0}[c, \lambda e_n ]^{(a)} \xi_a^{\dag}\\
&+[c, \lambda e_n ]^{(a)}(\omega - \omega_0)_a L_{\xi}^{\omega_0} c^{\dag}\\
=& L_{\xi}^{\omega_0} ([c, \lambda e_n ])^{(a)}(\xi_a^{\dag}- (\omega - \omega_0)_a c^\dag)-L_{\xi}^{\omega_0} ([c, \lambda e_n ]^{(a)})(\xi_a^{\dag}- (\omega - \omega_0)_a c^\dag)\\
&+ L_{\xi}^{\omega_0}[c, \lambda e_n ]^{(a)} \xi_a^{\dag}+[c, \lambda e_n ]^{(a)}(L_{\xi}^{\omega_0} (\omega - \omega_0)_a) c^\dag+[c, \lambda e_n ]^{(a)}(\omega - \omega_0)_a L_{\xi}^{\omega_0} c^{\dag}\\
=& L_{\xi}^{\omega_0} ([c, \lambda e_n ])^{(a)}(\xi_a^{\dag}- (\omega - \omega_0)_a c^\dag) +L_{\xi}^{\omega_0} ([c, \lambda e_n ]^{(a)})(\omega - \omega_0)_a c^\dag\\
&+[c, \lambda e_n ]^{(a)}(L_{\xi}^{\omega_0} (\omega - \omega_0)_a) c^\dag+[c, \lambda e_n ]^{(a)}(\omega - \omega_0)_a L_{\xi}^{\omega_0} c^{\dag}\\
=&  L_{\xi}^{\omega_0} ([c, \lambda e_n ])^{(a)}(\xi_a^{\dag}- (\omega - \omega_0)_a c^\dag) 
\end{align*}
This last term cancels out with the one resulting from the first computation.
\item \reft{S1S1g}{S1S1g8}, \reft{S1S1g}{S1S1g22} and \reft{S0S1f}{S0S1f4}:
\begin{align*}
 &-[L_{\xi}^{\omega_0}c, \lambda e_n ]^{(n)}\lambda^\dag + [c, L_{\xi}^{\omega_0} (\lambda e_n)^{(n)} e_n ]^{(n)}\lambda^{\dag}+L_{\xi}^{\omega_0} (\lambda e_n)^{(b)}([c,e])_b^{(n)}\lambda^\dag\\
 =& -(L_{\xi}^{\omega_0}[c, \lambda e_n ])^{(n)}\lambda^\dag
\end{align*}
\reft{S1S1g}{S1S1g20}, \reft{S0S1f}{S0S1f6}: since $e_a^{(n)}=0$ we have
\begin{align*}
&L_{\xi}^{\omega_0} ([c, \lambda e_n ]^{(n)} e_n)^{(n)}\lambda^\dag+ [c, \lambda e_n ]^{(b)}(L_{\xi}^{\omega_0} e)_b^{(n)}\lambda^\dag \\
=&L_{\xi}^{\omega_0} ([c, \lambda e_n ]^{(n)} e_n)^{(n)}\lambda^\dag+ L_{\xi}^{\omega_0} ([c, \lambda e_n ]^{(a)} e_a)^{(n)}\lambda^\dag\\
 =& (L_{\xi}^{\omega_0}[c, \lambda e_n ])^{(n)}\lambda^\dag
\end{align*}

\item \reft{S1S1g}{S1S1g9}, \reft{S1S1g}{S1S1g25} and \reft{S0S1f}{S0S1f17}:
\begin{align*}
&-[[c, \lambda e_n ]^{(a)}(\omega - \omega_0)_a,c] c^{\dag}- [c, \lambda e_n ]^{(a)} {d_{\omega_0}}_a c c^{\dag} \\
=&- [c, \lambda e_n ]^{(a)} d_{(\omega - \omega_0)_a} c c^{\dag}= [c, \lambda e_n ]^{(a)} (d_{\omega} c)_a c^{\dag}.
\end{align*}
 \item \reft{S1S1g}{S1S1g13}, \reft{S1S1g}{S1S1g30} and \reft{S0S1f}{S0S1f18}: as before.
 \item \reft{S1S1g}{S1S1g11}, \reft{S1S1g}{S1S1g26} and \reft{S0S1f}{S0S1f9}: as before.
 \item \reft{S1S1g}{S1S1g12}, \reft{S1S1g}{S1S1g27} and \reft{S0S1f}{S0S1f10}: as before.
 \item \reft{S1S1g}{S1S1g15}, \reft{S1S1g}{S1S1g31} and \reft{S0S1f}{S0S1f11}: as before.
 \item \reft{S1S1g}{S1S1g16}, \reft{S1S1g}{S1S1g32} and \reft{S0S1f}{S0S1f12}: as before.
 \item \reft{S1S1g}{S1S1g36} and \reft{S1S1g}{S1S1g42} :
 \begin{align*}
\frac{1}{2}(\iota_{[\xi,\xi]} {d_{\omega_0}}(\lambda e_n))^{(a)}(\xi_a^{\dag}- (\omega - \omega_0)_a c^\dag)+\frac{1}{2} [\iota_{\xi}\iota_{\xi}F_{\omega_0}, \lambda e_n ]^{(a)}(\xi_a^{\dag}- (\omega - \omega_0)_a c^\dag)\\
= (L_{\xi}^{\omega_0}L_{\xi}^{\omega_0}(\lambda e_n))^{(a)}(\xi_a^{\dag}- (\omega - \omega_0)_a c^\dag)
 \end{align*}

 \reft{S1S1g}{S1S1g14} and \reft{S0S1f}{S0S1f20}:
 \begin{align*}
& -L_{\xi}^{\omega_0} (\lambda e_n)^{(a)}(\omega - \omega_0)_a L_{\xi}^{\omega_0} c^{\dag} - L_{\xi}^{\omega_0} (\lambda e_n)^{(a)} (L_{\xi}^{\omega_0} \omega)_a c^\dag \\
 =& L_{\xi}^{\omega_0} (L_{\xi}^{\omega_0} (\lambda e_n)^{(a)}(\omega - \omega_0)_a)c^{\dag} - L_{\xi}^{\omega_0} (\lambda e_n)^{(a)} (L_{\xi}^{\omega_0} (\omega - \omega_0)_a) c^\dag
\\& - L_{\xi}^{\omega_0} (\lambda e_n)^{(a)} \partial_a \xi^c \omega_c c^\dag\\
 =& L_{\xi}^{\omega_0} (L_{\xi}^{\omega_0} (\lambda e_n)^{(a)})(\omega - \omega_0)_a c^{\dag}- L_{\xi}^{\omega_0} (\lambda e_n)^{(a)} \partial_a \xi^c \omega_c c^\dag
 \end{align*}
 \reft{S1S1g}{S1S1g33}, \reft{S1S1g}{S1S1g34} and previous relation:
 \begin{align*}
  & L_{\xi}^{\omega_0} (\lambda e_n)^{(a)}(\partial_a \xi^b) \xi_b^{\dag}+ L_{\xi}^{\omega_0} (\lambda e_n)^{(a)}\partial_b (\xi^b \xi_a^{\dag}) + L_{\xi}^{\omega_0} (L_{\xi}^{\omega_0} (\lambda e_n)^{(a)})(\omega - \omega_0)_a c^{\dag}\\
  &- L_{\xi}^{\omega_0} (\lambda e_n)^{(a)} \partial_a \xi^c \omega_c c^\dag\\
   =& \left[ L_{\xi}^{\omega_0} (\lambda e_n)^{(b)}(\partial_b \xi^c e_c)^{(a)} -L_{\xi}^{\omega_0} (L_{\xi}^{\omega_0} (\lambda e_n)^{(b)})e_b^{(a)}\right](\xi_a^{\dag}- (\omega - \omega_0)_a  c^\dag)
\end{align*}  
\reft{S0S1f}{S0S1f7} and previous relation:
\begin{align*}
&\left[-L_{\xi}^{\omega_0} (\lambda e_n)^{(b)}(L_{\xi}^{\omega_0} e)_b^{(a)}+L_{\xi}^{\omega_0} (\lambda e_n)^{(b)}(\partial_b \xi^c e_c)^{(a)}\right] (\xi_a^{\dag}- (\omega - \omega_0)_a  c^\dag)\\
&-L_{\xi}^{\omega_0} (L_{\xi}^{\omega_0} (\lambda e_n)^{(b)})e_b^{(a)}(\xi_a^{\dag}- (\omega - \omega_0)_a  c^\dag)\\
&= -(L_{\xi}^{\omega_0} (L_{\xi}^{\omega_0} (\lambda e_n)^{(b)}e_b))^{(a)}(\xi_a^{\dag}- (\omega - \omega_0)_a  c^\dag)
\end{align*}
 \reft{S1S1g}{S1S1g23} and previous relation:
\begin{align*}
& \left[-L_{\xi}^{\omega_0} (L_{\xi}^{\omega_0} (\lambda e_n)^{(n)} e_n)^{(a)} -(L_{\xi}^{\omega_0} (L_{\xi}^{\omega_0} (\lambda e_n)^{(b)}e_b))^{(a)}\right](\xi_a^{\dag}- (\omega - \omega_0)_a  c^\dag)\\
&= -(L_{\xi}^{\omega_0}L_{\xi}^{\omega_0}(\lambda e_n))^{(a)}(\xi_a^{\dag}- (\omega - \omega_0)_a c^\dag)
\end{align*}
This last term cancels out with the one resulting from the first computation.
 \item  \reft{S1S1g}{S1S1g37}  and  \reft{S1S1g}{S1S1g43}:
 $$ \frac{1}{2}(\iota_{[\xi,\xi]} {d_{\omega_0}}(\lambda e_n))^{(n)}\lambda^{\dag}+ \frac{1}{2}[\iota_{\xi}\iota_{\xi}F_{\omega_0}, \lambda e_n ]^{(n)}\lambda^\dag= (L_{\xi}^{\omega_0}L_{\xi}^{\omega_0}(\lambda e_n))^{(n)}\lambda^{\dag}$$
 \reft{S1S1g}{S1S1g24} and \reft{S0S1f}{S0S1f8}:
 \begin{align*}
 &- L_{\xi}^{\omega_0} (L_{\xi}^{\omega_0} (\lambda e_n)^{(n)} e_n)^{(n)}\lambda^\dag-L_{\xi}^{\omega_0} (\lambda e_n)^{(b)}(L_{\xi}^{\omega_0} e)_b^{(n)}\lambda^\dag\\
 &= -L_{\xi}^{\omega_0} (L_{\xi}^{\omega_0} (\lambda e_n)^{(n)} e_n)^{(n)}\lambda^\dag- L_{\xi}^{\omega_0} (L_{\xi}^{\omega_0} (\lambda e_n)^{(b)} e_b)^{(n)}\lambda^\dag \\
 &= - (L_{\xi}^{\omega_0}L_{\xi}^{\omega_0}(\lambda e_n))^{(n)}\lambda^{\dag}
\end{align*}
This last term cancels out with the one resulting from the first computation.
\item \reft{S0S1f}{S0S1f13}, \reft{S0S1f}{S0S1f14}, \reft{S0S1f}{S0S1f15}, \reft{S0S1f}{S0S1f16}, \reft{S0S1f}{S0S1f21}, \reft{S0S1f}{S0S1f22}, \reft{S0S1f}{S0S1f25} and \reft{S0S1f}{S0S1f26}:
Everything vanishes since $\lambda\lambda=0$ and $e_n^{(b)}=0$.
\item  \reft{S1S1g}{S1S1g38} and  \reft{S1S1g}{S1S1g39}: 
\begin{align*}
&+\frac{1}{2}[\xi,\xi]^a(\partial_a \xi^b) \xi_b^{\dag}+\frac{1}{2} [\xi,\xi]^a\partial_b (\xi^b \xi_a^{\dag})= \xi^c \partial_c \xi^a (\partial_a \xi^b) \xi_b^{\dag}+ \xi^c \partial_c \xi^a \partial_b (\xi^b \xi_a^{\dag})\\
&= \partial_b (\xi^c \partial_c \xi^a  \xi^b \xi_a^{\dag}) -  \xi^c \partial_c \partial_b \xi^a  \xi^b \xi_a^{\dag}
\end{align*}
where the last term vanishes since it is symmetric and antisymmetric in the indexes $b$ and $c$.
\item  \reft{S1S1g}{S1S1g44} $= -$ \reft{S0S1f}{S0S1f23}.
\item  \reft{S1S1g}{S1S1g45} $= -$ \reft{S0S1f}{S0S1f24}.
\item  \reft{S1S1g}{S1S1g41} and  \reft{S1S1g}{S1S1g46}:
\begin{align*}
-\frac{1}{2}\iota_{\xi}\iota_{\xi}F_{\omega_0} L_{\xi}^{\omega_0} c^{\dag} + \frac{1}{2}\iota_{[\xi,\xi]}\iota_{\xi}F_{\omega_0} c^{\dag}= \frac{1}{2}d_{\omega_0} (\iota_{\xi}\iota_{\xi}F_{\omega_0}\iota_{\xi}c^{\dag})
\end{align*}
\end{itemize}

\sloppy
\printbibliography 
\end{document}